\documentclass[journal]{IEEEtran}
\ifCLASSINFOpdf
\else
\fi
\usepackage{amsfonts}
\usepackage{cite}
\usepackage{amsmath,amssymb,amsfonts, dsfont, float}
\usepackage[inkscapearea=page]{svg} 
\usepackage{algorithmic}
\usepackage{graphicx}
\usepackage{textcomp}
\usepackage{xcolor}
\usepackage{mathtools}
\usepackage{amsthm}
\newtheorem{assumption}{Assumption}
\newtheorem{theorem}{Theorem}
\newtheorem{lemma}{Lemma}
\newtheorem{corollary}{Corollary}

\begin{document}

%
\title{Multi-Target Detection for Cognitive MIMO Radar Networks}
%
%
%

\author{Nicholas~L.K.~Goradia,
        Harpreet~S.~Dhillon,~\IEEEmembership{Fellow,~IEEE,}
        and~R.~Michael~Buehrer,~\IEEEmembership{Fellow,~IEEE}
        
\thanks{The authors are with Wireless@VT, Department of ECE, Virginia Tech, Blacksburg, VA 24061,
USA. Email: \{nickgoradia1, hdhillon, rbuehrer\}@vt.edu. This work was supported, in part, by Lockheed Martin through its University Research program. Any opinions, findings, conclusions, or recommendations expressed in this material are those of the authors and do not necessarily reflect the views of the sponsor.}
}

\maketitle

\begin{abstract}
In this work, we develop centralized and decentralized signal fusion techniques for constant false alarm rate (CFAR) multi-target detection with a cognitive radar network in unknown noise and clutter distributions. Further, we first develop a detection statistic for co-located monostatic MIMO radar in unknown noise and clutter distributions which is asymptotically CFAR as the number of received pulses over all antennas grows large, and we provide conditions under which this detection statistic is valid. We leverage reinforcement learning (RL) for improved multi-target detection performance, where the radar learns likely target locations in a search area. These results are then generalized to the setting of cognitive radar networks, where radars collaborate to learn where targets are likely to appear in a search area. 

We show a fundamental tradeoff between the spatial and temporal domain for CFAR detection in unknown noise and clutter distributions; in other words, we show a tradeoff between the number of radar antennas and the number of temporal samples. We show the benefits and tradeoffs with centralized and decentralized detection with a network of cognitive radars.  


\end{abstract}


%
\IEEEpeerreviewmaketitle

\section{Introduction}
%
%
%
%
\IEEEPARstart{T}{he} concept of cognitive radar (CR) and cognitive radar networks (CRNs) was first introduced in \cite{Haykin2006},\cite{Haykin2005}, describing a system where radars continuously observe the environment and learn from these observations. CRNs are characterized by a central node which combines the individual radar outputs. CRNs can be comprised of many different types of radar configurations; those most relevant to this work are multistatic and co-located monostatic MIMO. While multistatic MIMO radar networks can exploit the spatial diversity of a target's radar cross section (RCS) \cite{Haimovich2008}, co-located monostatic MIMO radars offer superior parameter identifiability and flexibility in the beamform design \cite{Stoica2007}. In this paper, we focus on a CRN comprised of distributed co-located monostatic MIMO radars, leveraging both cognitive processing and spatial diversity in the target RCS. 

While CRNs (and CRs) are able to adapt to highly dynamic environments, the classical radar constant false alarm rate (CFAR) detection problem often makes either \textit{a priori} assumptions on the environment or collects secondary data to estimate unknown noise and clutter (which we call disturbance) effects. Making prior assumptions about the disturbance distribution works against the idea of sensing a dynamic environment because if the disturbance distribution changes, then there is an issue of model mismatch which can degrade detection performance. If no prior assumptions are made on the disturbance distribution, then secondary data will have to be collected to estimate it; however, this process will have to be repeated to keep up with general non-stationary processes which is inefficient. Furthermore, dyanmic environments generally contain multiple targets which should be detected. 

Motivated by these issues, we develop a Wald-type detection statistic which is robust to general, unknown disturbance distributions. We then generalize this detection statistic to radar networks using both centralized and decentralized signal fusion. We define centralized signal fusion by the fusion of each individual radar's entire received signal at the central node and decentralized signal fusion as the combination of individual radar's detection statistics at the central node. We show that by applying reinforcement learning (RL) to radar networks, multi-target detection performance can be greatly improved. 

\subsection{Related Work}

It has been shown that complex elliptically symmetric (CES) random variables, specifically, the subclass known as compound Gaussian (CG) random variables, are appropriate for modeling noise and clutter in general radar environments as the assumption of AWGN does not always hold \cite{Watts1987}, \cite{Ollila2012}. In general, detection statistics such as the popular generalized likelihood ratio test (GLRT) require assumptions be made on the disturbance distributions. Our goal is to perform detection in MIMO radar without knowing the disturbance distribution and without collecting secondary data.

CFAR detection for arbitrary non-Gaussian CES distributed disturbances has been studied through the lens of random matrix theory in works such as \cite{Kammoun2018}. CFAR detection with CR where the radar chooses between a library of distributions is explored in \cite{metcalf2015}; however, these methods rely on secondary data. CFAR detection, without secondary data, in a single snapshot with arbitrarily distributed disturbances was originally examined in \cite{greco2020} for the massive MIMO (MMIMO) case by leveraging asymptotic properties gained through the use of a large number of antennas. The methods in \cite{greco2020} have very high hardware costs associated as MMIMO will need a very large number of RF chains. Furthermore, this work only covers the case of monostatic radar. To date, there has been no generalization of this method to radar networks; extending it to a network of MMIMO radars would demand a prohibitive number of antennas and RF chains, resulting in high implementation costs. 

Furthermore, radar detection requires a search over a predetermined space. If this search is done by scanning one cell at a time, then this is clearly not the most efficient approach. If the entire area is searched at once; however, then the power of the transmitted beam is low and multi-target detection performance is poor. These issues are solved by learning target patterns to efficiently search multiple cells at once. This is originally proposed in \cite{Greco2021} where a state-action-reward-state-action (SARSA) RL algorithm is paired with the detection statistic developed in \cite{greco2020}. The $\epsilon-$greedy policy in \cite{Greco2021}, is later improved on in \cite{greco2022} to solve the issue of missed target recovery by dynamically adapting hyper-parameters in the SARSA algorithm. The work in \cite{wang2024} further improves on the missed target aspect of detection with SARSA by using Bayesian priors to assist in angle bin selection during the beamforming step; however, this paper only considers AWGN. All of these RL methods consider monostatic MMIMO radar, whereas we wish to use RL for a network of radars.

The area of CRNs allows for the exploration of many problems such as spectrum sensing, allocation and waveform selection which have been studied in works such as \cite{Tony2018}, \cite{howard2021},\cite{howard2023}. We do not make any considerations on resource allocation or spectrum sharing in our work as we consider each radar to operate at a separate frequency to avoid interference; however, our work does consider waveform selection in CRNs. Our work differs in the sense that we do not make waveform selection under a set of resource allocation constraints, rather, our waveform selection problem is motivated by beamforming for multi-target detection. Waveform design in CRNs is considered in \cite{Rossetti2018}, however, this work considers multistatic radar networks and robust waveform design for clutter. In our work, we consider a detection statistic which is robust to unknown clutter independent of waveform design (as long as the transmitted signals are orthonormal). 

One can imagine that a network of radars, working collaboratively to sense the environment and learn, will outperform a single radar; however, CRNs introduce their own challenges such as synchronization, resource allocation, node placement, and determining how to combine the individual radar outputs. We focus on the issue of output combination which can be subdivided into two cases: centralized and decentralized fusion. 
    
    In terms of signal fusion for radar networks, there has been much work exploring GLRTs which combine information from multiple radars. In \cite{Varshney1986}, the optimal fusion rule --under decentralized detection with individual radars making hard, binary decisions-- is developed as a function of the probabilities of hypothesis $0$ and $1$. In \cite{Varshney1996}, detection for distributed radar networks is considered where Neyman-Pearson tests are analyzed. Specifically, it is shown that the Neyman-Pearson test is optimal in decentralized detection such that each individual radar performs the uniformly most powerful test. Our detection statistics differ in the sense that we consider Wald-type tests rather than Neyman-Pearson tests due to parameter estimation where the Neyman-Pearson test would not be optimal. The general problem of distributed signal fusion in sensor networks under the assumption of conditional independence for each sensor in the network is further explored in: \cite{Reibman1987} where optimal decentralized detection is derived and performance is shown for non-Gaussian but known noise distributions, and the authors of \cite{thomopoulos1987} further consider Neyman-Pearson tests for decision fusion in sensor networks.

    There are also various works which explore decision fusion under correlated observations between the sensors in the network such as \cite{Lauer1982} which explores this under white noise, and \cite{drakopolous} and \cite{kam1992} discuss a Neyman-Pearson test for binary decisions at the individual sensors.
    
    Furthermore, in \cite{kay2013}, performance loss in centralized against decentralized detection with the GLRT is analyzed. In \cite{Mrstik1978}, ``M by N" detection for radar networks is introduced. We do not consider this type of decentralized detection as it offers worse performance than the chosen decentralized fusion described in Sec.~\ref{sec:dec_det}. Detection in passive MIMO radar networks is considered in \cite{Hack2014}.


\subsection{Contributions}
To summarize, we consider a network of distributed co-located monostatic MIMO radars which are working collaboratively to perform multi-target detection in an unknown disturbance distribution. The key contributions of this paper are listed as follows. 
\subsubsection{Robust detection statistic for MIMO radar} We develop a CFAR detection statistic which is robust to general, unknown disturbance processes for co-located MIMO radar. We accomplish this by leveraging the spatio-temporal domain, and as such, we are able to accomplish radar detection in unknown disturbances using significantly fewer antennas than existing work in this area. 

\subsubsection{Outline necessary conditions for robust CFAR detection to be valid} We outline five necessary conditions which are, as a whole, unique to our detection statistic. We describe why these necessary conditions are reasonable to hold in general. 

\subsubsection{Modeling of heavy-tailed 2-D disturbance processes} Since the disturbance process in our setting is, in general, 2-D in nature, we carefully consider modeling and CFAR detection in this regime. Because autoregressive (AR) processes are known to capture heavy tailed characteristics of various clutter processes, we show that there exists at least a class of 2-D AR processes under which our detection statistic is CFAR. 
\subsubsection{Develop centralized and decentralized robust CFAR detection statistics for CRNs} We develop both centralized and decentralized detection statistics which are CFAR and robust to general unknown disturbance distributions for CRNs. We describe the benefits and drawbacks for both detection methods, and we derive the associated probabilities of false alarm and detection. 

\subsubsection{Reinforcement Learning for CR and CRNs} We outline the state, action, and reward for RL to be applied to multi-target detection for CRNs. We keep this section general so that any RL algorithm can be applied to the multi-target detection problem. 

\subsubsection{Design Insights} We show multi-target detection performance with various algorithms for monostatic radar in unknown disturbance processes. We also show multi-target detection performance for various radar network scenarios. We determine that by using a network of radars, detection performance is significantly improved, and finally, we show that RL algorithms improve multi-target detection performance in lower SNR settings such that a network of many radars is able to improve multi-target detection performance without the need for RL. 

\subsection{Notation}
We denote some variable $x$ associated with the $i$th radar in a network as $x^{\{i\}}$ which is raised to the $n$th power as $(x^{\{i\}})^n$. Note that in general, a variable $x$ raised to the $n$th power is denoted as $x^n$. We use this superscript notation to avoid overloading subscript indices and to avoid confusion with exponentiation. We use lowercase $\mathbf{v}$ and uppercase $\mathbf{A}$ bold lettering to denote vectors and matrices, respectively. The notation $\mathbf{A}^T$ and $\mathbf{A}^H$ denotes transpose and conjugate transpose. The conjugate of $a$ is written as $a^*$. The trace of a matrix $\mathbf{A}$ is written as ${\rm tr}\{\mathbf{A}\}$. The Kronecker product is written as $\otimes$. Statistical expectation is represented by $\mathbb{E}[\cdot]$. The $i,j$th entry of a matrix is denoted by $[\cdot]_{i,j}$. The indicator function is written as $\mathds{1}(\cdot)$. The function ${\rm mod}(a,b)$ means $a\mod b$. Given some matrix $\mathbf{X}\in \mathbb{F}^{N\times N}$ where $\mathbb{F^{N\times N}}$ is some arbitrary field, the function $\rm vec(\cdot)$ is defined as ${\rm vec}(\mathbf{X}) = [x_{11}, x_{21}, \dots x_{N1}, \dots x_{NN}]^T$. Given some collection of matrices $\mathbf{X}_1,\dots,\mathbf{X}_M$ such that $\mathbf{X}_i \in \mathbb{F}^{N \times N}$ for all $i = 1,\dots,M$, we define 
\begin{equation*}
    {\rm blkdiag}(\mathbf{X}_1, \dots, \mathbf{X}_M) = 
    \begin{bmatrix}
        \mathbf{X}_1 & \mathbf{0} & \dots & \mathbf{0} \\
        \mathbf{0} & \mathbf{X}_2 & \dots & \mathbf{0} \\
        \vdots & \vdots & \ddots & \vdots \\
        \mathbf{0} & \mathbf{0} & \dots & \mathbf{X}_M
    \end{bmatrix} \in \mathbb{F}^{NM \times NM}.
\end{equation*} 
For some real valued function $f(x)$ and positive real valued function $g(x)$, $f(x) = \mathcal{O}(g(x))$ means that $\exists a \in \mathbb{R}_{\geq 0}, x_0 \in \mathbb{R}$ such that $|f(x)| \leq ag(x), \quad \forall x\geq x_0$.


\section{Signal Model}
We first build the signal model for the monostatic MIMO setting. Once this foundation is built, we generalize the signal model to a network of radars. 
\subsection{Monostatic Radar}\label{sec:monostatic_radar}
Consider a co-located MIMO radar with $\rm M_R$ receive and $\rm N_T$ transmit antennas. We denote the transmit array manifold vector as $\mathbf{a}(\cdot)$ and similarly the array manifold vector as $\mathbf{b}(\cdot)$. Let $\alpha$ represent the deterministic and unknown two-way pathloss and RCS constant. Note that this value is kept deterministic and unknown such that it is general and not adhering to any specific model. It is a function of the target location; however, we will not denote this relationship explicitly for ease of notation. For a target located at angle $\phi$ and delay $\tau$, the complex baseband received signal can be written as \eqref{eqn:rxsig}, \cite{1Friedlander2012}, \cite{2Friedlander2012}.
\begin{equation}\label{eqn:rxsig}
    \mathbf{x}(t) = \alpha \mathbf{b}(\phi)\mathbf{a}^T(\phi)\mathbf{s}(t-\tau)e^{j\omega t} + \mathbf{n}(t), \quad t \in [0,T],   
\end{equation}
where $\tau$ and $\omega$ represent the time delay and Doppler shift induced by the target.  

For $\rm N_T$ independent, orthonormal signals $\mathbf{u}(t) \in \mathbb{C}^{\rm N_T}$, the transmitted signal is given as $\mathbf{s}(t) = \mathbf{W}\mathbf{u}(t)$ for some beamforming weight matrix $\mathbf{W}$ adhering to power constraint ${\rm tr}\{\mathbf{WW}^H\} = P_{\rm T}$ for total transmitted power $P_{\rm T}$ \cite{2Friedlander2012}. Note that $\mathbf{s}(t) \in \mathbb{C}^{\rm N_T \times 1}$ and $\mathbf{W} \in \mathbb{C}^{\rm N_T \times N_T}$. Due to the use of orthonormal transmit signals, matched filtering results in a virtually larger received signal of size  $\rm N_TM_R$ \cite{Stoica2007}, \cite{1Friedlander2012}. 

More specifically, we match filter the received signal using the set of orthonormal transmit signals $\mathbf{u}(t)$, resulting in ${\rm N_T}$ matched filters tuned to $\bar{\tau}$ and $\bar{\omega}$ at each receive antenna. The output of the matched filters is written as \eqref{eqn:mf} \cite{1Friedlander2012},

\begin{align}\label{eqn:mf}
    \mathbf{X}(\bar{\tau}, \bar{\omega}) &= \int_{0}^{T}\mathbf{x}(t)\mathbf{u}^H(t-\bar{\tau})e^{-j \bar{\omega}t} {\rm d}t\nonumber \\
     &= \alpha\mathbf{b}(\phi)\mathbf{a}^T(\phi)\mathbf{W}\nonumber\\
    &\times \int_{0}^{T}\mathbf{u}(t-\tau)\mathbf{u}^H(t-\bar{\tau})e^{-j(\bar{\omega} - \omega)t}{\rm d}t \nonumber \\
    &+ \int_0^{T}\mathbf{n}(t-\tau)\mathbf{u}^H(t-\bar{\tau})e^{-j\bar{\omega}t}{\rm d}t \nonumber \\
    &=\alpha\mathbf{b}(\phi)\mathbf{a}^T(\phi)\mathbf{W}\nonumber\\
    &\times \int_{0}^{T}\mathbf{u}(t-\tau)\mathbf{u}^H(t-\bar{\tau})e^{-j(\bar{\omega} - \omega)t}{\rm d}t +  \mathbf{C}(\bar{\tau},\bar{\omega}), 
\end{align}
where $\mathbf{C}(\cdot,\cdot)$ represents the disturbance. If we assume the filter is matched to the delay $\tau$ and Doppler $\omega$ \cite{Greco2021}, we can write the matched filter output due to a single PRI as \eqref{eqn:rxmf}

\begin{equation}\label{eqn:rxmf}
    \mathbf{x} = {\rm vec}(\mathbf{X}(\tau, \omega)) = \alpha \mathbf{v} + \mathbf{c},
\end{equation}
where $\mathbf{c} \in \mathbb{C}^{\rm N_TM_R}$ is the disturbance and 
$\mathbf{v} = \left(\mathbf{W}^T\mathbf{a}(\phi)\right) \otimes \mathbf{b}(\phi)$ \cite{Greco2021}.

We assume that the spatial search space is divided into $l = 1, \dots, L$ angle bins and that \eqref{eqn:rxmf} is processed by a bank of spatial filters tuned to a specific angle range as in \cite{Greco2021}; the received signal for the $l$th angle bin is given as
\begin{equation}\label{eqn:rxmfangle}
    \mathbf{x}_l = \alpha_l\left(\mathbf{W}^T\mathbf{a}(\phi_l)\right)\otimes\mathbf{b}(\phi_l) + \mathbf{c}_l.
\end{equation}
For each angle bin, consider a collection of $K$ received signal pulses (corresponding to one PRI each) to be one coherent processing interval (CPI). We drop the subscript $l$ for ease of notation and write the collection of $K$ pulses as $\tilde{\mathbf{X}} = \left[\mathbf{x}[1], \dots, \mathbf{x}[K]\right]$ where $\tilde{\mathbf{X}} \in \mathbb{C}^{\rm{M_R N_T} \times K}$. We then vectorize $\tilde{\mathbf{X}}$ as 
\begin{align}\label{eqn:rxvec}
   {\rm vec}(\tilde{\mathbf{X}}) = \tilde{\mathbf{x}} 
   = [x_1[1], \dots, x_N[1], \dots x_1[K], \dots, x_N[K]],
\end{align}
where $x_n[k]$ in \eqref{eqn:rxvec} represents the signal at the $n$th spatial channel and $k$th pulse. We can then rewrite $\tilde{\mathbf{x}}$ in the form of \eqref{eqn:rxmf}. Define $\tilde{\mathbf{W}} = {\rm blkdiag}\left(\mathbf{W}, \dots, \mathbf{W}\right) \in \mathbb{C}^{\rm N_T K \times N_T K}$, and $\tilde{\mathbf{a}}(\phi) = [{\mathbf{a}}(\phi)^T, \dots, {\mathbf{a}}(\phi)^T]^T \in \mathbb{C}^{\rm N_T K \times 1}$ then we can rewrite \eqref{eqn:rxmfangle} as 
\begin{equation}\label{eqn:blckdiag}
    \tilde{\mathbf{x}} = \alpha\left(\tilde{\mathbf{W}}^T \tilde{\mathbf{a}}(\phi)\right) \otimes \mathbf{b}(\phi) + \tilde{\mathbf{c}};
\end{equation}
for each angle bin, the received signal is given as,
\begin{equation}\label{rxvec2}
    \tilde{\mathbf{x}}_l = \alpha_l \tilde{\mathbf{v}}_l + \tilde{\mathbf{c}}_l,
\end{equation}
where $\tilde{\mathbf{v}}_l$ and $\tilde{\mathbf{c}}_l$ are arranged as in \eqref{eqn:blckdiag} and \eqref{eqn:rxvec}.

\subsection{Radar Networks}
We now describe the extension to radar networks. Consider a network of $R$ radars. Each radar in the network will follow the signal model outlined in Sec.~\ref{sec:monostatic_radar}. Suppose for some reference radar in the network, which has its view partitioned into $L$ angle bins, that there exists a bijective mapping between the set of angle bins for all other radars and the reference radar. In other words, for the $l$th angle bin at reference radar $i$, any other radar $j \neq i$ has a unique angle bin $q$ associated only with bin $l$ of the $i$th radar. We assume that each radar in the network is a co-located MIMO radar which transmits on a different frequency band to avoid interference. Furthermore, we assume perfect time synchronization between radars, and a noise and fading free communications channel between the radars and the central processing node. The received signal at the $i$th radar during the $p$th CPI is written as 
\begin{equation}
    \tilde{\mathbf{x}}^{\{i\}}[p] = \alpha^{\{i\}}[p]\tilde{\mathbf{v}}^{\{i\}}[p] + \tilde{\mathbf{c}}^{\{i\}}[p].
\end{equation}

\subsection{Assumptions}\label{sec:Assumptions}
In order to formulate our detection problem, we make the following assumptions which are to be justified at the end of this section. 

\begin{assumption}\label{assumption1}
    The radar cross section $\alpha^{\{i\}}$ remains constant over some number of pulses $K$ for all $i=1,\dots,R$ radars.
\end{assumption}
\begin{assumption}\label{assumption2}
    The disturbance distribution remain stationary over some number of pulses $K$ for all radars.
\end{assumption}
\begin{assumption}\label{assumption3}
    Within the search space, the number of targets, their corresponding angle bins, and SNR remain unchanged over $K$ pulses.
\end{assumption}
\begin{assumption}\label{assumption4}
    The disturbance process seen by each radar has an associated autocorrelation function which decays at least polynomially. This allows us to estimate the covariance matrix as sparse. Formally, following the main assumption in \cite{greco2020}, let $\{\tilde{c}_n\}$ be a discrete second order stationary process of the disturbances defined in \eqref{rxvec2}. Then $\forall n, \quad \mathbb{E}[\tilde{c}_n\tilde{c}_{n-m}^*] = \mathcal{O}(|m|^{-\gamma})$, $m\in\mathbb{Z}$ and $\gamma > \rho/(\rho-1)$ for some $\rho > 1$.
\end{assumption}
\begin{assumption}\label{assumption5}
    The disturbance process received by one radar is independent of the disturbance process received by all other radars. 
\end{assumption}

Note that these assumptions are reasonable for targets whose position and orientation are not changing over $K$ pulses. Furthermore, the final assumption is justified as we assume the radars in the network are distributed to exploit SNR gains from spatial diversity. Assumption~\ref{assumption4} puts a constraint on the autocorrelation properties of the disturbance process. Specifically, we require that the process has an autocorrelation function which decays at least polynomially in time and in space. This property certainly is reasonable in the spatial domain as works such as \cite{greco2020} have used this to develop a detection statistic; however, one might ask if this is reasonable temporally. We discuss this assumption for the temporal domain in Sec.~\ref{sec:cluttermodels}. Furthermore, general heavy-tailed stochastic processes can be modeled by an AR process. We show, in Lemma~\ref{lemma1}, that the vectorized 2D-AR process follows Assumption~\ref{assumption4}.

    \subsection{Clutter Models}\label{sec:cluttermodels}
    As radar resolution capabilities have increased, it has been shown that radar clutter follows non-Gaussian distributions. Both land and sea clutter have been modeled using distributions such as the K-distribution \cite{rosenberg2012},\cite{sayama2001}. Furthermore, it has been shown that the temporal correlation of clutter in windy vegetation environments can be modeled by the intrinsic clutter motion model \cite{Billingsley1997}. This model has a Doppler spectrum with exponential decay which means that its autocorrelation (Fourier transform) follows polynomial decay. Furthermore, for sufficiently dynamic conditions (related to wind speeds and wave height), sea clutter has been described as having temporal autocorrelation which decays at least in polynomial time \cite{rosenberg2012}. It can be concluded that there are clutter models for both land and sea clutter which have temporal autocorrelation functions that decay at least in polynomial time, in line with Assumption~\ref{assumption4}.

\section{Cognitive Radar for Detection}
We now introduce the definition of state, action, and reward which will be used to apply reinforcement learning for multi-target detection. The overall goal is to point power into bins where there are likely targets. 

\subsection{Monostatic Radar}
Assuming that our signal is processed in $K$ pulse batches (one CPI is equivalent to $K$ PRIs) we define a time step of such batches as $p$ such that the test statistic in one angle bin $l$ over the $p$th $K$ number of pulses (i.e. the $p$th CPI) is denoted as $\Lambda_{p,l}$.

\subsubsection{State Space}\label{sec:state_space}
We define the state space as $\mathcal{S} = \{0,1,\dots,T_{\rm max}\}$ where $T_{\rm max}$ is the maximum number of targets in the environment. The state space is the set of the number of possible targets in the environment. The state at time instant $p$ --written with some overload of notation-- $s_p$ is the number of detections made across all angle bins at CPI $p$,
\begin{equation}\label{eqn:state}
    s_p = \min \left\{\sum_{l=0}^{L-1}{\mathds{1}(\Lambda_{p,l} \geq \lambda)}, T_{\rm max} \right\},
\end{equation}
where $\lambda$ is the detection threshold. 
\subsubsection{Reward}\label{sec:Reward}
Let $\hat{P}_{{\rm D},l}$ be the estimated probability of detection for some angle bin $l$, and let $\mathcal{L}$ be the set of angles associated with the state $s_p$. The reward at time $p$ is given by 
\begin{equation}\label{eqn:reward}
    r_p = \sum_{l\in\mathcal{L}}\hat{P}_{{\rm D},l} - \sum_{l\notin\mathcal{ L}}\hat{P}_{{\rm D},l}.
\end{equation}

\subsubsection{Action Space}
Let $B_{\max}$ be the maximum number of bins that we point power into at once. Since we have a maximum number of targets $T_{\max}$, we choose $B_{\max}=T_{\max}$ as we wouldn't point power into directions where we know there won't be any targets. The action space $\mathcal{A}$ is defined as $\mathcal{A} =\{0,1,\dots,B_{\max}\}$. 
Note that the action space is mathematically equivalent to the state space; however, the action is not necessarily the same as the state. More intuitively, the state represents the number of targets that we have detected, and the action represents the number of bins which the algorithm points the transmit beam into.

\subsection{Cognitive Radar Networks}
Each radar in the network will have their own action; however, they will all share a common state, reward, and detection statistic $\boldsymbol{\Lambda}$ which is created by fusing information from each individual radar at the central node. The common statistic is organized according to the angle bins of the reference radar. The bijective mapping between angle bins of radars is essential so that each element of a radar's individual detection statistic can be associated with an element of the reference radar's detection statistic. 

Because the reward is calculated based on the estimated probability of detection, the estimated probability of detection is calculated as described in Theorem.~\ref{thm:Pd_decentralized} for decentralized fusion and with the standard Marcum Q function \cite{nuttall1975} for centralized fusion detection.

\section{Main Results}
We develop the key results first for the case of a single CR, then using these key results, we develop centralized and decentralized fusion methods for CRNs. 
\subsection{CFAR Detection for Monostatic MIMO}
\subsubsection{Modeling Noise and Clutter} \label{sec:modeling}
An unknown disturbance process --which is complex discrete with continuous power spectral density-- has second order statistics which can be approximated by an AR process \cite{stoica2005spectral}, \cite{greco2020}. Furthermore, the AR process can capture properties of heavy tailed processes \cite{greco2020}. We show that there exist at least a class of AR processes which adhere to Assumption~\ref{assumption4}, such that this assumption reasonably holds for arbitrary heavy-tailed distributions. Furthermore, we will show that we can perform CFAR detection with MIMO in unknown noise and clutter distributions which have the property outlined in Assumption~\ref{assumption4}.

We present the following lemma to support the main theorem of this section. 
\begin{lemma}\label{lemma1}
    Consider a 2-D AR($p,q$) noise process defined as follows,
    \begin{equation}\label{eqn:2dAR}
        {c}[n,k] = \sum_{i=1}^{p}\sum_{j=1}^{q}\left[\mathbf{\Phi}\right]_{i,j}{c}[n-i,k-j] + \epsilon[n,k],
    \end{equation}
    for $\mathbf{c}\in \mathbb{C}^{N\times K}$ and where $\epsilon[\cdot,\cdot]$ are zero mean i.i.d. innovations; furthermore, let \eqref{eqn:2dAR} have an associated analytic minimum phase filter \cite{Marzetta1978ALP}, then the autocorrelation function $\tilde{\rho}[r]$ associated with $\tilde{\mathbf{c}} = {\rm vec}\{\mathbf{c}\}$ decays exponentially.
\end{lemma}
\begin{proof}
    In order to properly define a mapping between elements (or indices) of $\mathbf{c}$ and $\tilde{\mathbf{c}}$, let us consider $n\in [1,N], k\in [1,K]$. By \cite[Theorem 4.1(b)]{Marzetta1978ALP} and \cite{Ranganath1985}, we have that the autocorrelation of \eqref{eqn:2dAR} is bounded exponentially as
    \begin{equation}\label{eqn:2DBound}
        \rho[n,k] \leq A \gamma_1^{\lvert n\rvert}\gamma_2^{\lvert k\rvert},
    \end{equation}
    for some constant $A \in \mathbb{R}$ and $\gamma_1,\gamma_2 \in (0,1)$ where the Spectral Density Function of \eqref{eqn:2dAR} denoted by $S(z_1,z_2)$ is analytic in $\{\gamma_1 < \lvert z_1\rvert < 1/\gamma_1,\gamma_2 < \lvert z_2\rvert < 1/\gamma_2\}$. We then have that, from the vectorization of \eqref{eqn:2DBound}, 
    \begin{align}
        \tilde{\rho}[r] &\leq A\gamma_1^{\lvert n\rvert}\gamma_2^{\lvert k\rvert} \nonumber \\
        &=A\gamma_1^{\lvert 1+{\rm mod}\left((r-1),NK\right)\rvert}\gamma_2^{\lvert 1 + \lfloor \frac{r-1}{NK}\rfloor\rvert} \nonumber \\
        &= A\gamma_1^{\lvert r-NK\lfloor\frac{r-1}{NK} \rfloor\rvert}\gamma_2^{\lvert 1 + \lfloor \frac{r-1}{NK}\rfloor\rvert} \nonumber \\
        & \leq A \left[\left(\gamma_1\gamma_2\right)^{\frac{1}{NK}}\right]^{\lvert r \rvert},
    \end{align}
    so we have that the autocorrelation function of the vectorized 2-D AR process decays exponentially. 
\end{proof}
Note that the above lemma verifies the existence of a subset of 2-D AR($p,q$) processes which have vectorized autocorrelation functions which decay exponentially, and that this is a sufficient condition. We make no claim towards the derivation of such processes. 

\subsubsection{CFAR Detection Statistic}
The following theorem introduces the CFAR detection statistic used. 
\begin{theorem}\label{theorem1}
    Let $N = {\rm M_R N_T}K$. A robust Wald-type detector can be constructed as 
    \begin{equation}\label{eqn:thm1}
        \Lambda = \frac{2|\tilde{\mathbf{v}}^H\tilde{\mathbf{x}}|^2}{\tilde{\mathbf{v}}^H \hat{\mathbf{\Gamma}} \tilde{\mathbf{v}}},
    \end{equation}
    where $\hat{\mathbf{\Gamma}}$ is the estimated covariance matrix of the disturbances $\tilde{\mathbf{c}}$ calculated by \eqref{gamma_hat}. As $N$ approaches infinity, the distribution of $\Lambda$ approaches a chi-squared distribution under both $H_0$ and $H_1$.
\end{theorem}
\begin{proof}
By Assumption~\ref{assumption4}, Assumption~\ref{assumption1}, Assumption~\ref{assumption2} and Assumption~\ref{assumption3}, Theorem~\ref{theorem1} follows directly from \cite{greco2020}. To summarize the results of \cite{greco2020}, the estimate of $\alpha$ is $\sqrt{N}$ consistent and Gaussian distributed. The detection statistic can then be written as the multiplication of two asymptotically normally distributed random variables. Thus, the detection statistic is asymptotically chi-squared distributed.
\end{proof}

From Theorem~\ref{theorem1}, we are now able to construct a detection statistic which is CFAR asymptotically as $N$ grows large. This detection statistic is CFAR under generally unknown disturbance distributions as described in Corollary~\ref{corollary1}.

\begin{corollary}\label{corollary1}
    For a disturbance process defined in \eqref{eqn:2dAR}, the detection statistic defined in \eqref{eqn:thm1} of Theorem~\ref{theorem1} is asymptotically chi-squared distributed under both $H_0$ and $H_1$.
\end{corollary}
\begin{proof}
    From Lemma~\ref{lemma1}, any vectorized 2-D AR process with an associated minimum phase filter satisfies Assumption~\ref{assumption4}. Under all other assumptions, the requirements for Theorem~\ref{theorem1} are fulfilled and the result follows. 
\end{proof}

\subsubsection{Estimation of Disturbance Covariance Matrix}
We estimate the disturbance covariance matrix as in \cite{greco2020}; for some truncation lag $l$, disturbance vector of length $N$ and $\hat{c}_n = \tilde{x}_n - \hat{\alpha}\tilde{v}_n$ as

\begin{align}\label{gamma_hat}
    [\hat{\mathbf{\Gamma}}]_{i,j} = 
    \begin{cases}
        \hat{c}_i \hat{c}_j^* & 0 \leq j-i \leq l \\
        \hat{c}_i^* \hat{c}_j & 0 \leq i-j \leq l \\
        0 & |i-j| > l    
    \end{cases}.
\end{align}
\subsection{Fusion in Cognitive Radar Networks}
We now know how to perform CFAR detection for individual co-located MIMO radar in unknown disturbance processes. We explore centralized and decentralized fusion methods to determine a detection statistic which uses information from all radars in the network.
\subsubsection{Centralized Detection}
Centralized detection is characterized by each radar in the network sending an entire received signal to a central node as opposed to just some soft decision. Of course transmitting a much larger amount of data to the central node comes with challenges but also with the benefits of improved detection performance. While we provide no analysis of the challenges introduced by communicating larger data sets from individual radars to the central node, this problem is promising for future lines of work. 

We wish to combine our signals in some ``optimal" way. One natural option is maximal ratio combining. 
\begin{theorem}\label{thm:mrc}
    Consider $R$ received signals $\{\mathbf{x}^{\{i\}}\}_{i=1}^R$, such that $\mathbf{x}^{\{i\}} = \alpha^{\{i\}} \mathbf{v}^{\{i\}} + \mathbf{c}^{\{i\}}$ where $\mathbf{c}^{\{i\}}$ is zero mean random process with covariance $\Gamma^{\{i\}}$, and $\mathbf{c}^{\{i\}}$ is independent of $\mathbf{c}^{\{j\}}$ for $i\neq j$. The MRC combination is given as 
    \begin{equation}\label{eqn:centralized_signal}
        \sum_{i=1}^R\left(w^{\{i\}}\right)^*\mathbf{x}^{\{i\}},
    \end{equation}
    where $w^{\{i\}} = \alpha^{\{i\}}$.
\end{theorem}
\begin{proof}
    Consider the instantaneous SNR at time index $j$, where we drop indices for convenience, as 
    \begin{equation}
        \frac{\mathbb{E}\left[\left\lvert \sum_{i=1}^R\left(w^{\{i\}}\right)^*\alpha^{\{i\}} v^{\{i\}} \right\rvert^2\right]}{\mathbb{E}\left[\left\lvert \sum_{i=1}^R \left(w^{\{i\}}\right)^* c^{\{i\}} \right\rvert^2\right]} = \frac{\left\lvert \sum_{i=1}^R\left(w^{\{i\}}\right)^*\alpha^{\{i\}} v^{\{i\}} \right\rvert^2}{\sum_{i=1}^R \left\lvert w^{\{i\}}\right\rvert^2 [\Gamma^{\{i\}}]_{jj}},
    \end{equation}
    where $[\Gamma^{\{i\}}]_{jj}$ corresponds to the $j$th diagonal entry of the covariance matrix associated with $\mathbf{c}^{\{i\}}$. By the Cauchy-Schwartz inequality, we then have 
    \begin{equation}
        \frac{\left\lvert \sum_{i=1}^R\left(w^{\{i\}}\right)^*\alpha^{\{i\}} v^{\{i\}} \right\rvert^2}{\sum_{i=1}^R \lvert w^{\{i\}}\rvert^2 [\Gamma^{\{i\}}]_{jj}} \leq \sum_{i=1}^R\frac{\lvert \alpha^{\{i\}} v^{\{i\}} \rvert^2}{[\Gamma^{\{i\}}]_{jj}},
    \end{equation}
    where the equality is met only if $w^{\{i\}} = \alpha^{\{i\}}$. 
\end{proof}
Note that each disturbance process is zero mean, and the new disturbance process is a weighted sum of the disturbance of each radar. The summation of these disturbances will have an autocorrelation function which decays exponentially, so we do not violate Assumption~\ref{assumption4} for this centralized detection method. 

\subsubsection{Decentralized Detection}\label{sec:dec_det}
For decentralized detection, each radar creates a soft detection statistic which is then sent to the central node for fusion. We organize each radar's detection statistic vector such that the statistic relating to each angle bin is organized according to the reference radar's statistic. We then take the maximum detection statistic for each bin between all the radars. The detection statistic is computed as,
\begin{equation}\label{eqn:decentralized_statistic}
    {\Lambda}_{p,l} = \max\left\{{\Lambda}_{p,l}^{\{i\}}\right\}_{i=1}^{R}.
\end{equation}

Now the family of chi-squared distributions is certainly not closed under the $\max(\cdot)$ operator, so we must compute the distribution of the new statistic for determining the CFAR threshold. 

\begin{theorem}\label{thm:PFA}
The probability of false alarm for decentralized detection using $\max$ fusion is given by 
\begin{equation}
    P_{\rm FA} = \int_{\lambda}^{\infty}R\left(F_{\Lambda^{\{i\}}\rvert \mathcal{H}_0}(x)\right)^{R-1}f_{\Lambda^{\{i\}}\rvert \mathcal{H}_0}(x){\rm d}x,
\end{equation}
where $f_{\Lambda^i\rvert \mathcal{H}_0}(x)$ is the pdf of the $i$th radar's detection statistic under the null hypothesis, and $R$ is the number of radars. 
\end{theorem}
\begin{proof}
    For a set of $n$ i.i.d. random variables with pdf $f(\cdot)$ and cdf $F(\cdot)$, the pdf of the $k$th order statistic is given by 
    \begin{equation}
        \frac{n!}{(k-1)!(n-k)!}(F(x))^{(n-1)}(1-F(x))^{(n-k)}f(x).
    \end{equation}
    Plugging in the pdf of each radar's detection statistic under the null hypothesis and replacing $n$ by $R$ and $k$ by $R$, the result is straightforward. 
\end{proof}

Because we also use the estimated probability of detection in the RL algorithm for detection, we must consider this new function. 
\begin{theorem}\label{thm:Pd_decentralized}
    The probability of detection for decentralized detection using $\max$ fusion is given by 
    \begin{equation}
        P_{\rm D} = 1-\prod_{i=1}^{R}\left[1-Q_1\left(\sqrt{\Lambda^{\{i\}}}, \sqrt\lambda\right)\right],
    \end{equation}
    where $Q_1(\cdot, \cdot)$ is the Marcum Q function of the first order.
\end{theorem}
\begin{proof}
    The probability of detection for a single radar is given by $Q_1(\sqrt{\Lambda}, \sqrt{\lambda})$ which is the complementary CDF of the non-central chi-squared distribution \cite{nuttall1975}. By definition, this is $\mathbb{P}[\Lambda \geq \lambda]$; however, we consider 
    \begin{align}
        P_{\rm D} &= \mathbb{P}[\max\{\Lambda^{\{1\}}, \dots, \Lambda^{\{R\}}\} \geq \lambda] \nonumber \\
        &= 1-\mathbb{P}[\Lambda^{\{1\}} < \lambda] \cdot {\dots} \cdot \mathbb{P}[\Lambda^{\{R\}} < \lambda] \nonumber \\
        &= 1-\prod_{i=1}^R\left[1-Q_1\left(\sqrt{\Lambda^{\{i\}}}, \sqrt\lambda\right)\right],
    \end{align}
    where the last step follows from the definition of $Q_1(\cdot,\cdot)$.
\end{proof}

\section{Simulation Results}
We present simulation results for co-located MIMO radar with various detection algorithms in an unknown disturbance distribution for different dynamic scenarios. We then provide performance of CRNs for both decentralized and centralized detection and compare the two scenarios. 

\subsection{Formulation of Disturbance Process}
We choose an AR(6,6) process of the form outlined in Lemma~\ref{lemma1} with complex t-distributed innovations defined in \cite{Ollila2012}, \cite{greco2020} and correlation parameters given in \cite{greco2020}. 
\subsection{Simulation Parameters}
We use the following parameters described by table~\ref{tab:radarparams} in all simulations discussed. By choosing 20 angle bins and 10 transmit antennas, the beampattern will be 3dB down in adjacent bins. Specifically, the 3dB beamwidth is approximately $1/N_T$ in the spatial frequency domain, given half wavelength antenna spacing. Furthermore, the search space of 180 degrees is partitioned into 20 angle bins. If the number of angle bins were increased, then the beamwidth would leak into adjacent bins. This isn't necessarily an issue as on receive, the spatial filter will allow us to create a detection statistic for our intended angle bin; however, with increased number of angle bins, the RL algorithms will have to search many more bins and convergence will take longer on average. If we decrease the total number of angle bins and the total number of targets in the search space is larger than the total number of angle bins, then all targets will not be able to be detected; furthermore, if the angle bins are widened such that the 3dB beamwidth within one angle bin only covers a small portion, then there could be an issue of missed targets when targets are not located in the center of the angle bins. The number of transmit antennas and total angle bins should be chosen such that the search space is discretized to detect a sufficient number of targets at once while maintaining a 3dB beamwidth that at least covers the intended bins. 

\begin{table}[H]
\caption{Radar Specifications}
\label{tab:radarparams}
\centering
\begin{tabular}{||c c||} 
 \hline
 Parameter & MIMO  \\ [0.5ex] 
 \hline\hline
 ${\rm N_T}$ & 10 \\
\hline
 ${\rm M_R}$ & 10\\ 
 \hline
 Pulses per CPI & 100 \\ 
 \hline
Angle bins $L$ & 20 \\
\hline
 $\kappa$ & 0.8 \\ 
 \hline
 Monte-Carlo Trials & 400 \\ 
 \hline
 $P_{\rm FA}$ & $1\times10^{-4}$ \\
 \hline
 PRI & 5 $\mu$s \\
 \hline

\end{tabular}
\end{table}

\subsection{Detection Algorithms}
We consider five detection algorithms for comparison. The first is the ``Optimal" algorithm. In this case, we consider that the radar has perfect knowledge of where all targets are located at all discrete time steps. For each pulse, the radar points its beams into the bins where the targets are known to be located. Note that in this scenario, the detection probability is not always unity as multiple targets with lower SNR values can cause detection performance to degrade, even when beams are pointed only into bins where targets are located. 

Next, we consider RL based algorithms. Because we do not consider the disturbance process to be known \textit{a priori}, we use a model free RL algorithm. We use the on-learning ``SARSA" algorithm developed in \cite{greco2022} as non-optimal actions are punished more compared to off-learning. We also show the Bayesian assisted SARSA algorithm developed in \cite{wang2024} which is denoted as ``Bayes". 

We then consider two non-RL based algorithms. These are termed ``Adaptive" and ``Orthogonal". In the Adaptive algorithm, the radar points its beams into the bins where a detection was last made. If no detections are made, it defaults to the Orthogonal algorithm. The Orthogonal algorithm puts equal power into all angle bins. 

Finally, we consider an ``Scanning" algorithm which points power into a single bin per CPI and sweeps through each bin meaning that during the next CPI, the algorithm points power into the adjacent bin. We consider this algorithm only for the first two scenarios as we highlight the RL-based algorithm's superiority. 

\subsection{Spatial and Temporal Domain Trade-offs}
    We will discuss the differences between the extreme case of MMIMO and a single PRI versus many PRI with much fewer antennas. In the MMIMO regime, there are many more antennas which require more RF chains; as such, the receive gain is much greater in this regime. When the total transmit power across one CPI is the same in the MMIMO and MIMO regimes, performance will be better in the MMIMO regime due to the additional receive gain; however, if the transmit power per pulse is equal in both regimes, then performance in the MIMO regime will be superior due to larger effective transmit power due to additional pulses. If transmit power is scaled such that the norm of \eqref{rxvec2} is equivalent in both regimes, then of course performance will also be equivalent. 

\subsection{Co-located MIMO Radar Simulation Scenarios}
We now outline the two different dynamic simulation scenarios which are used to analyze co-located MIMO radar detection performance in arbitrarily distributed disturbances. For the entirety of this section, one discrete time instance is equivalent to one CPI. 

\subsubsection{Scenario 1}
Consider 400 discrete time steps which corresponds to 2 seconds. Two targets begin in the scene in angle bins 5 and 13 with SNR -18 and -21 dB, respectively. At time instant 100, the target in bin 5 disappears, but the target in bin 13 remains with the same SNR. At time 200, a target appears in bin 17 with SNR -20 dB, and the target in bin 13 remains with SNR -21 dB. Finally, at time 300, the target in bin 13 disappears, and the target in bin 17 remains with SNR -20 dB. 

We can see that the orthogonal and adaptive algorithms have poor detection performance for all three targets in Fig.~\ref{fig:s1:T1}, Fig.~\ref{fig:s1:T2}, and Fig.~\ref{fig:s1:T3}. In fact, the orthogonal algorithm fails to detect any targets and the adaptive algorithm fails to get above 50\% detection performance. This is due to the fact that the adaptive algorithm begins with an orthogonal waveform, and when the target SNR is too low, there will never be a detection. Furthermore, if the adaptive algorithm were to begin by pointing its power into only one or a few directions, then any targets outside of those directions will not be detected. The final non-RL based algorithm --the Scanning algorithm-- has excellent detection performance when it points in the direction of the target. 

In contrast, all of the RL based algorithms have excellent detection performance. When a new target is introduced, the RL algorithms adjust with performance around 90\% after 20 CPIs. If one were to scan the search area, it would take 20 CPIs to scan all angle bins, but this comes with the draw back of not being able to detect multiple targets in one CPI. As a metric for comparison, we can measure the probability of target acquisition, where acquisition is defined as a target being detected $M$ out of $N$ times. 

The probability of acquisition will depend on the beamwidth and the choice of $M$ and $N$. If we require $3$ detections before acquisition is declared, then the Scanning algorithm will need $60$ CPIs if the transmit beamwidth is constrained to a single bin; however, if the transmit beamwidth leaks into adjacent bins, then the Scanning algorithm could make $3$ detections in consecutive CPIs.
Specifically, we require that a every $5$ CPI, if a target is detected $3$ times, then it is deemed acquired. Because the transmit beam has its $3$dB down point located in the center of adjacent bins, the Scanning algorithm is essentially able to detect $3$ targets at once: those in the bin its pointing into and those in adjacent bins. We can see some of the benefits of an Scanning algorithm when the RL based algorithms are slower to converge in Fig.~\ref{fig:Pa:T3}. Note that if the requirement for $M$ is larger than 3, then the beamwidth would have to be much wider on the Scanning algorithm to meet the detection requirements in less than $20$ CPIs where 20 is the number of angle bins; this would lead to the issue of lowering overall power into the desired bin, and as $M$ approaches $20$, the beam would approach the orthogonal algorithm which has poor detection performance. The RL based algorithms do not suffer as $M$ increases for some fixed $N$, however. 

The probability of false alarm performance is as expected for all algorithms, and a more detailed discussion is presented in Table~\ref{tab:PFA}. 

\begin{figure}[H]
    \centering
    \includegraphics[width=0.95\linewidth]{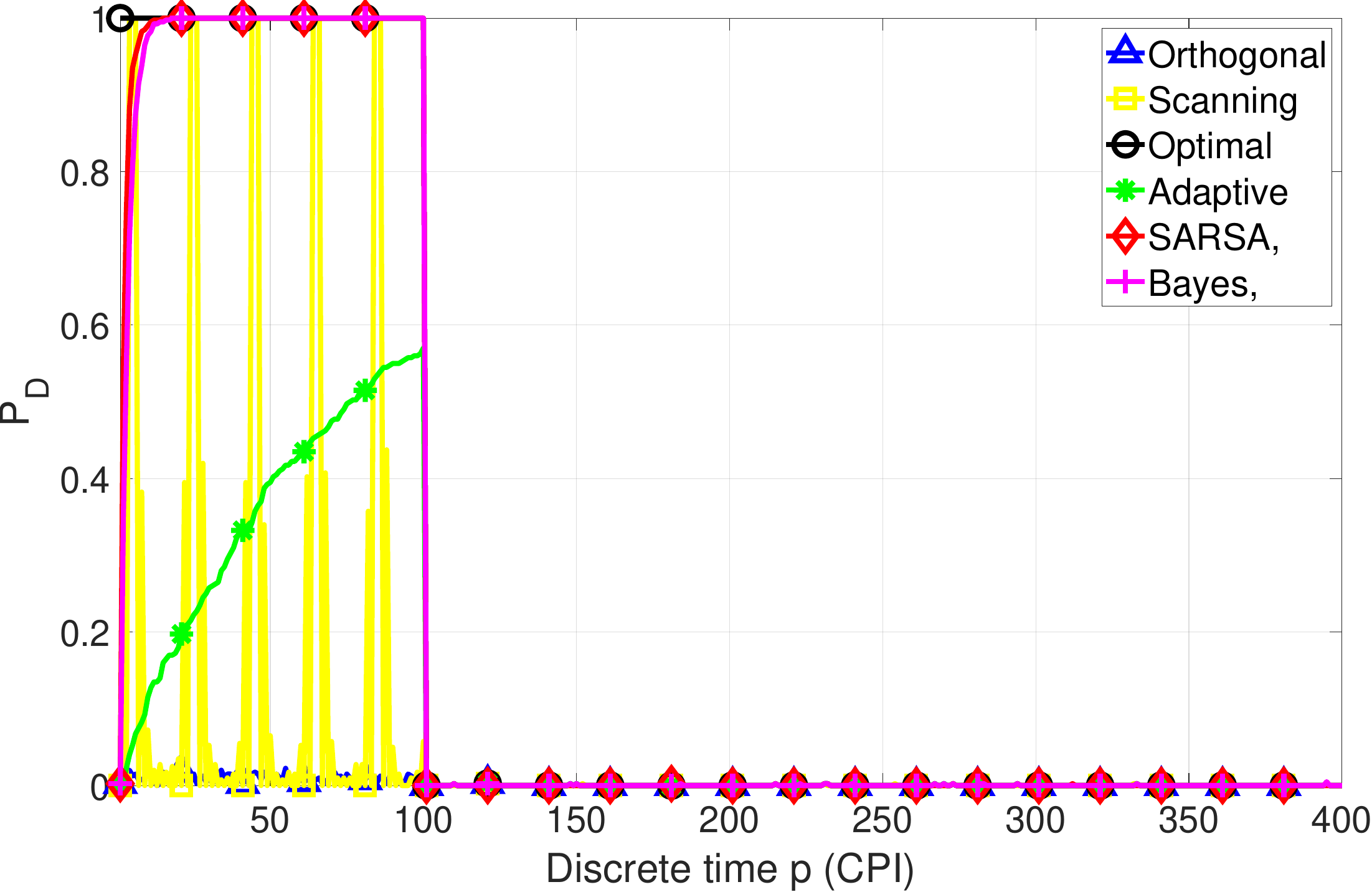}
    \caption{Scenario 1: MIMO (${\rm M_RN_T = 100}$) Detection Performance for the Target in Bin 5}
    \label{fig:s1:T1}
\end{figure}
\begin{figure}[H]
    \centering
    \includegraphics[width=0.95\linewidth]{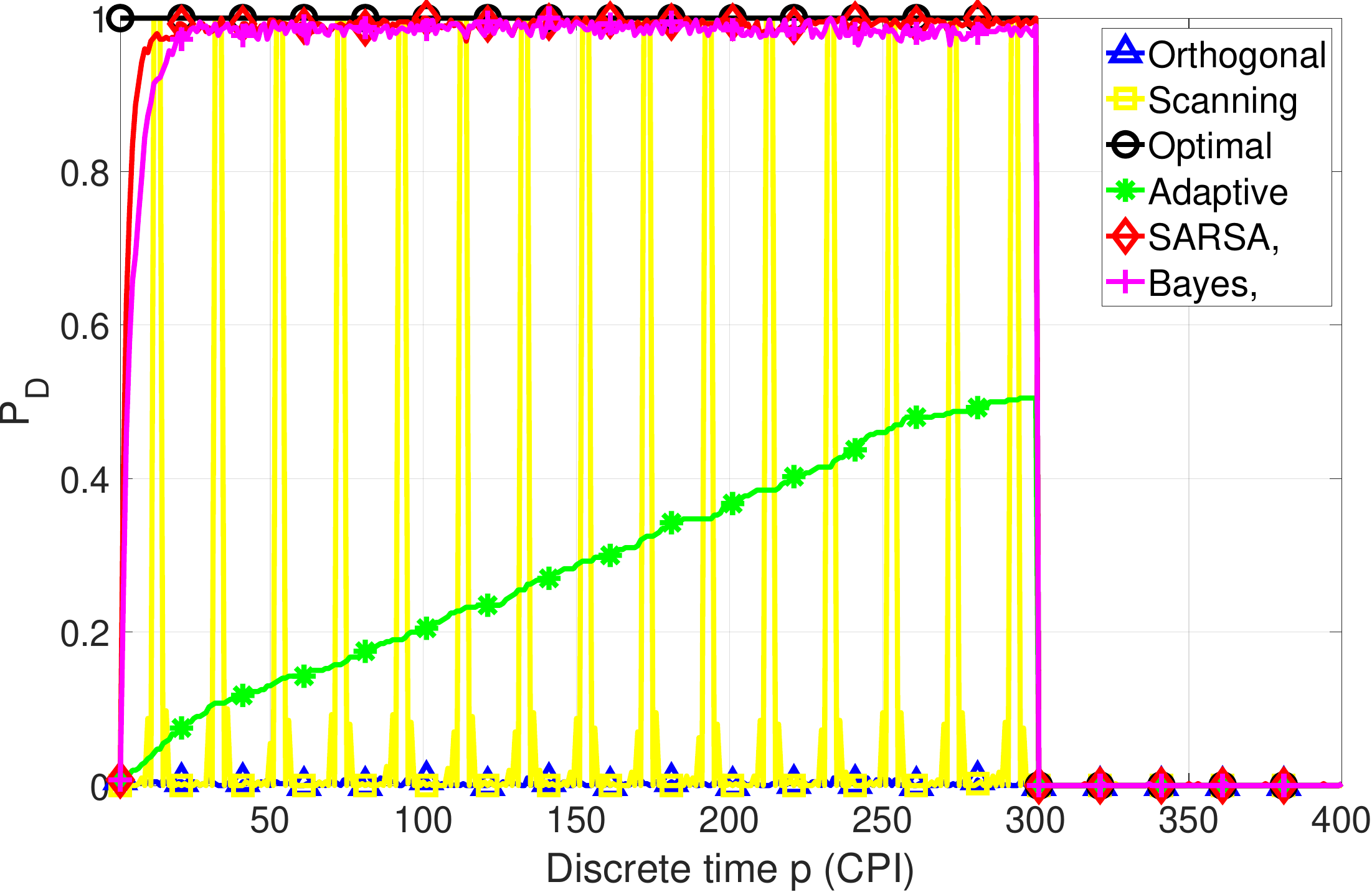}
    \caption{Scenario 1: MIMO (${\rm M_RN_T = 100}$) Detection Performance for the Target in Bin 13}
    \label{fig:s1:T2}
\end{figure}
\begin{figure}[H]
    \centering
    \includegraphics[width=0.95\linewidth]{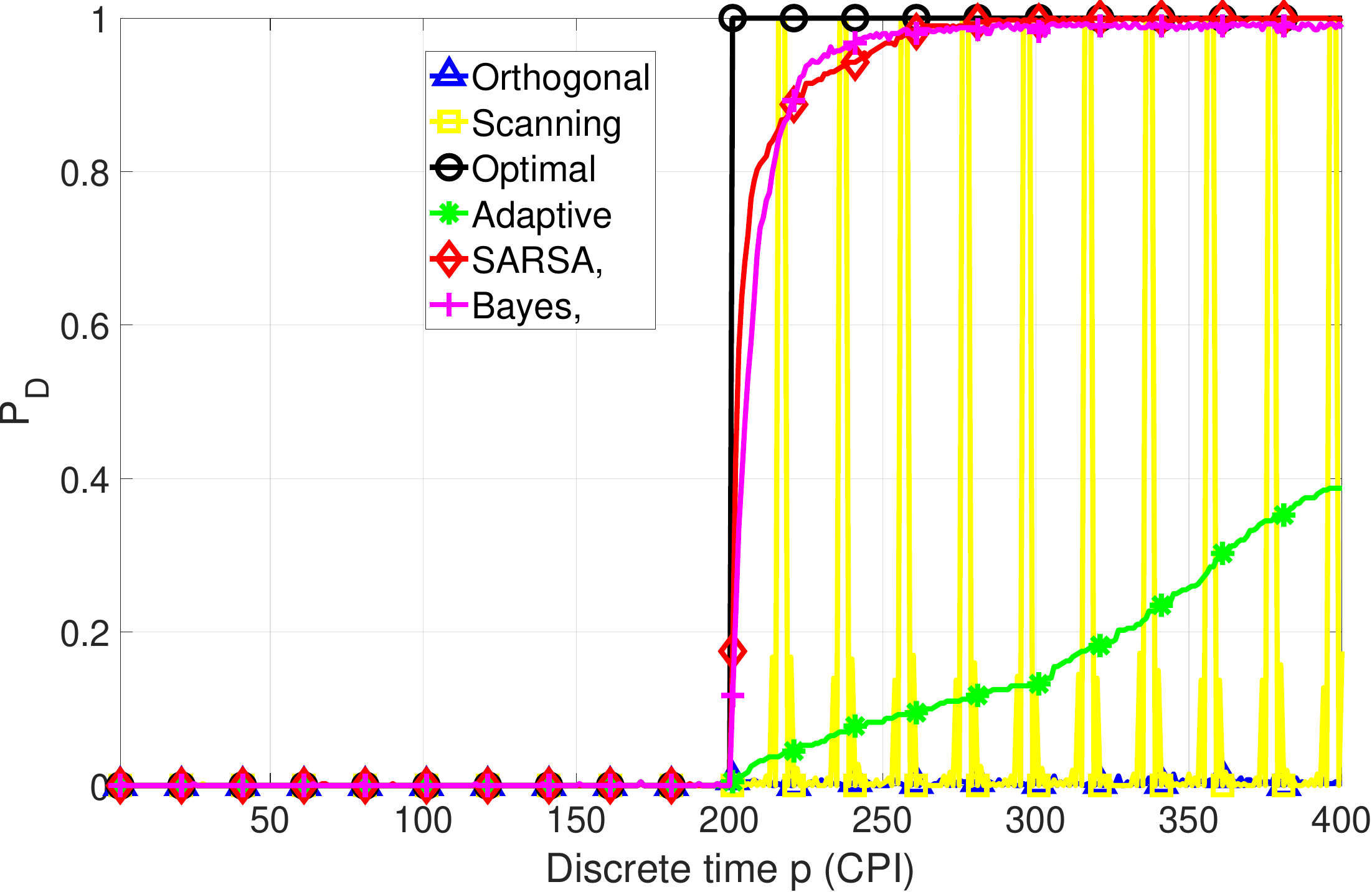}
    \caption{Scenario 1: MIMO (${\rm M_RN_T = 100}$) Detection Performance for the Target in Bin 17}
    \label{fig:s1:T3}
\end{figure}

\begin{figure}[H]
    \centering
    \includegraphics[width=0.95\linewidth]{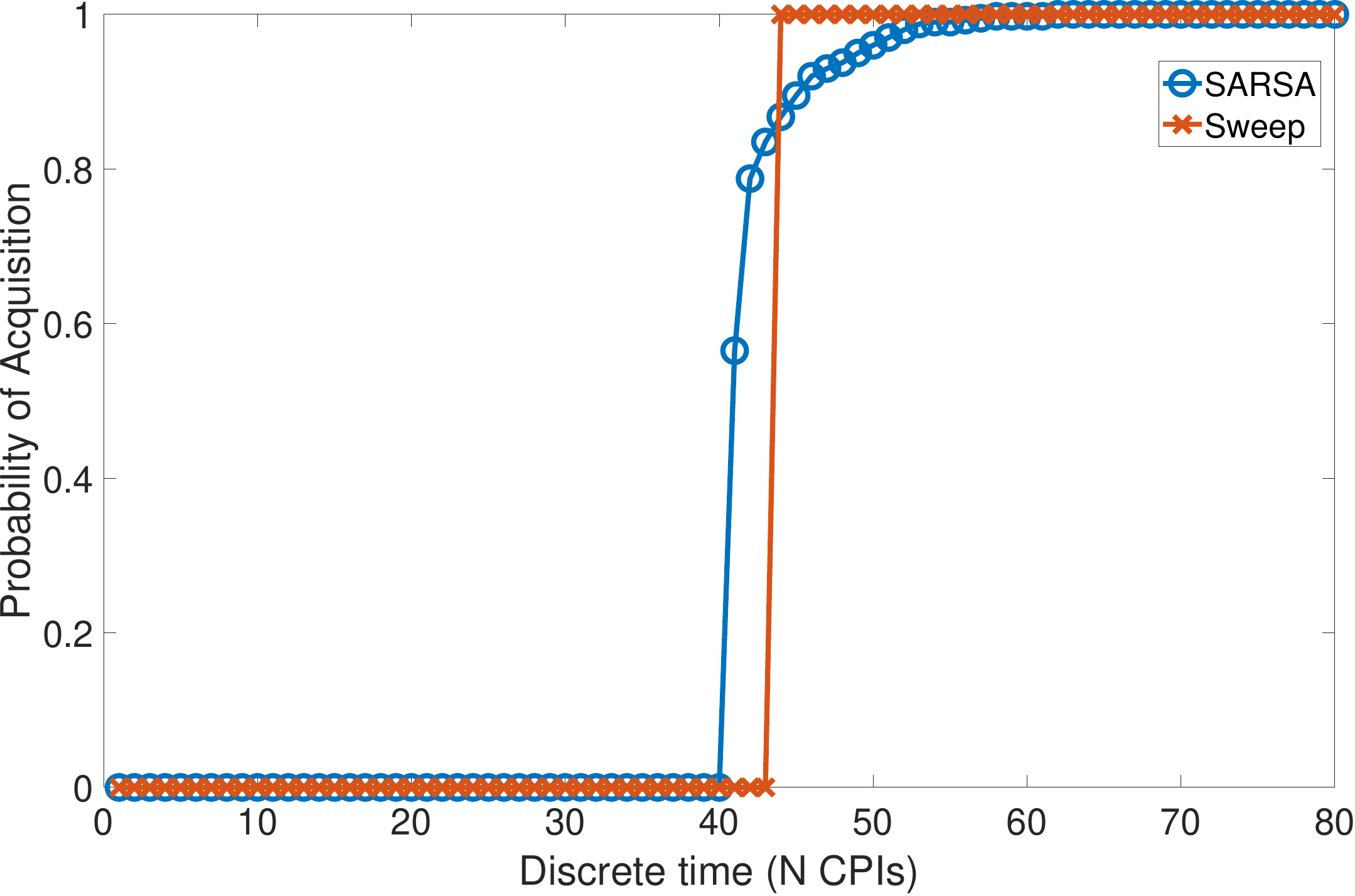}
    \caption{Scenario 1: MIMO (${\rm M_RN_T = 100}$) Probability of Acquisition for the Target in Bin 17}
    \label{fig:Pa:T3}
\end{figure}

\subsubsection{Scenario 2}
Now consider a simulation which lasts for 200 discrete time steps. There is a target in angle bin 7 and 16. Both of the targets follow an SNR curve described by Fig.~\ref{fig:s2:SNR}. 
\begin{figure}[H]
    \centering
    \includegraphics[width=0.95\linewidth]{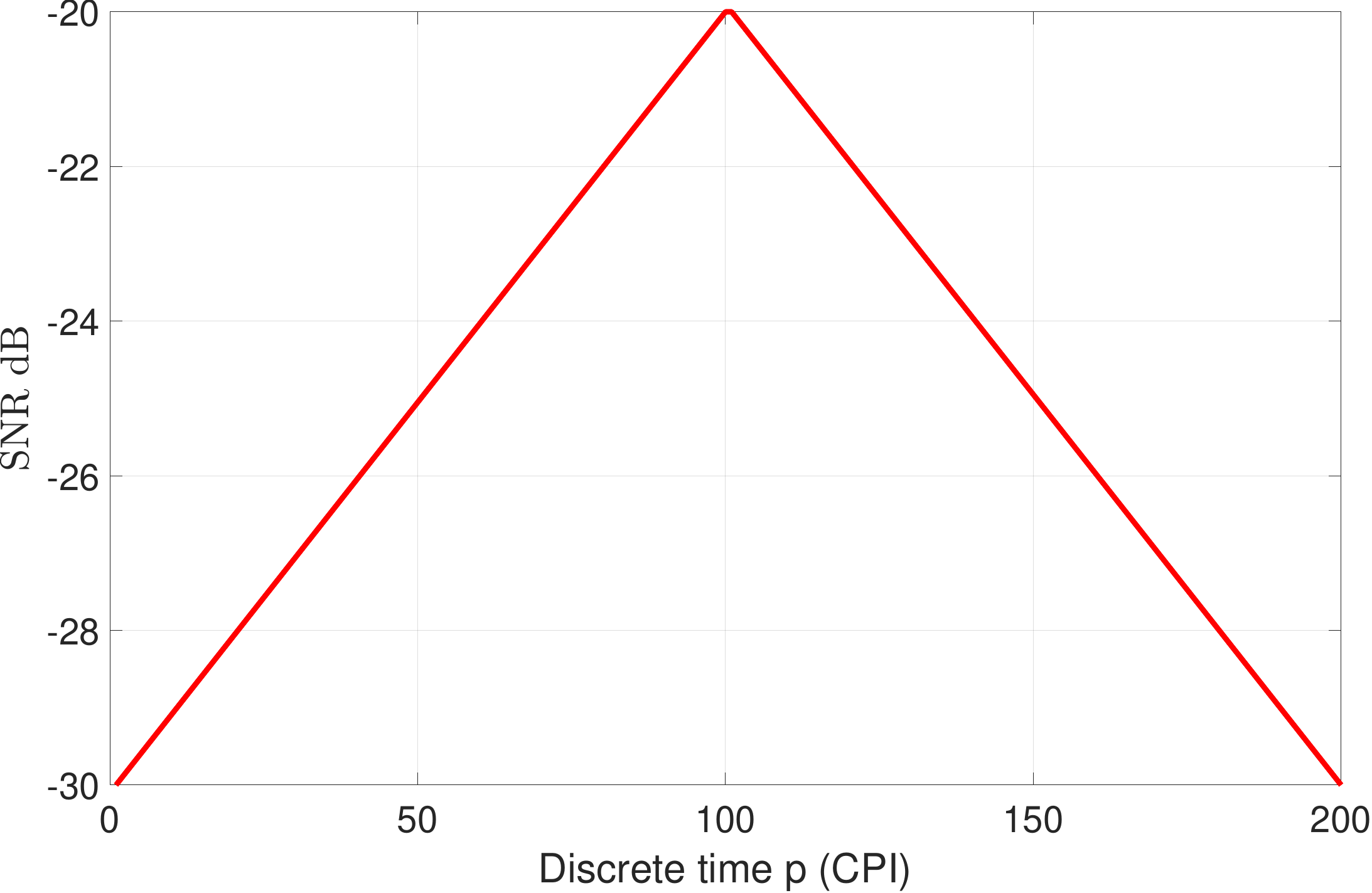}
    \caption{SNR Seen By Targets in Scenario 2}
    \label{fig:s2:SNR}
\end{figure}

We can see from Fig.~\ref{fig:s2:MIMOT1} and Fig.~\ref{fig:s2:MIMOT2} that detection performance follows the SNR curve in Fig.~\ref{fig:s2:SNR}. We can also see that detection performance for both targets is jointly optimized, where if detection performance were to increase for one target, then it would decrease for the other. Again the RL based algorithms significantly outperform the adaptive and orthogonal algorithms. While it seems like the Scanning algorithm outperforms the RL based algorithms while SNR is low, we can look at the probability of acquisition shown in Fig.~\ref{fig:Pa:T1} where we can see that the SARSA algorithm has a higher probability of acquisition for $M=3$ and $N=5$.

We also present the probability of false alarm for this scenario in Table~\ref{tab:PFA} and see all the values are approximately the set value of $1\times10^{-4}$. Note that because the disturbance distribution is the same as in Scenario 1 and because the thresholds are not changing, the probability of false alarm performance is equivalent between the two scenarios. 

\begin{table}[H]
\caption{Measured Probability of False Alarm for Scenario 2}
\label{tab:PFA}
\centering
\begin{tabular}{||c c||} 
 \hline
 Algorithm & $P_{\rm FA}$  \\ [0.5ex] 
 \hline\hline
 Orthogonal & $1.0064 \times 10^{-4}$ \\
\hline
 Optimal & $1 \times10^{-4}$\\ 
 \hline
 Adaptive & $9.5139 \times 10^{-5}$ \\ 
 \hline
SARSA & $9.8611 \times 10^{-5}$ \\
\hline
 Bayes & $9.7222 \times10^{-5}$ \\ 
 \hline

\end{tabular}
\end{table}

\begin{figure}[H]
    \centering
    \includegraphics[width=0.95\linewidth]{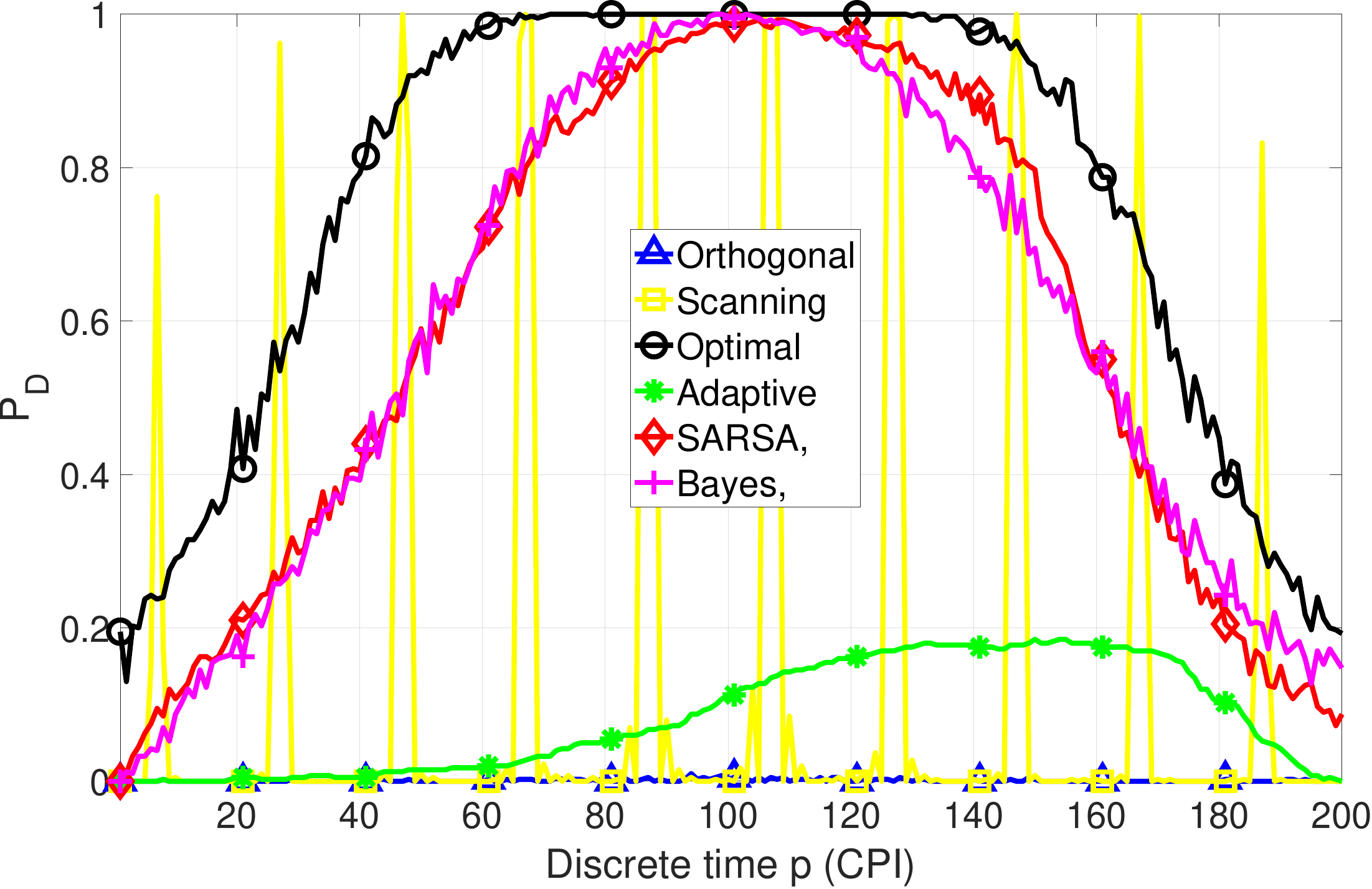}
    \caption{Scenario 2: MIMO (${\rm M_RN_T = 100}$) Detection Performance for the Target in Bin 3}
    \label{fig:s2:MIMOT1}
\end{figure}
\begin{figure}[H]
    \centering
    \includegraphics[width=0.95\linewidth]{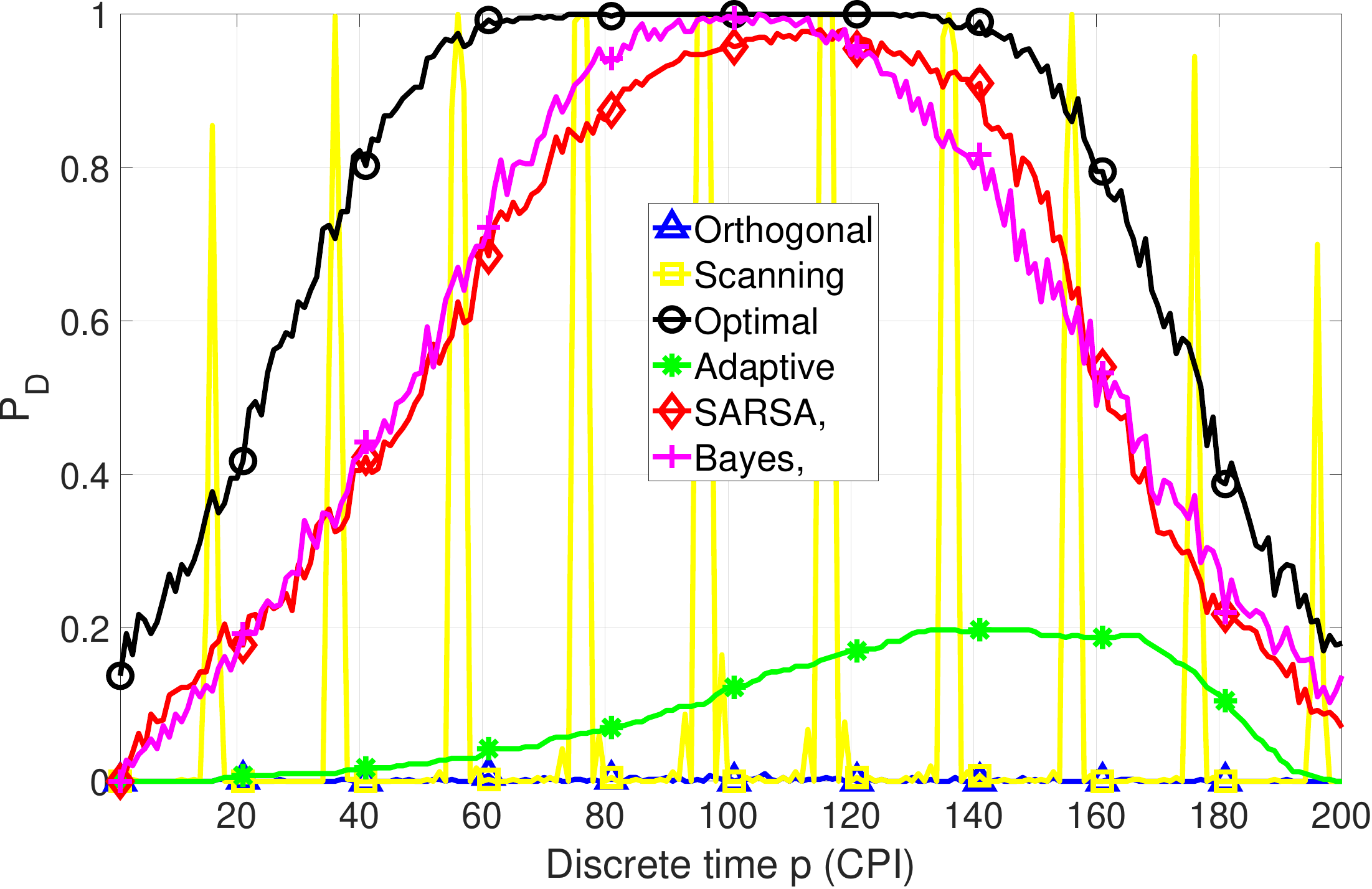}
    \caption{Scenario 2: MIMO (${\rm M_RN_T = 100}$) Detection Performance for the Target in Bin 18}
    \label{fig:s2:MIMOT2}
\end{figure}
\begin{figure}[H]
    \centering
    \includegraphics[width=0.95\linewidth]{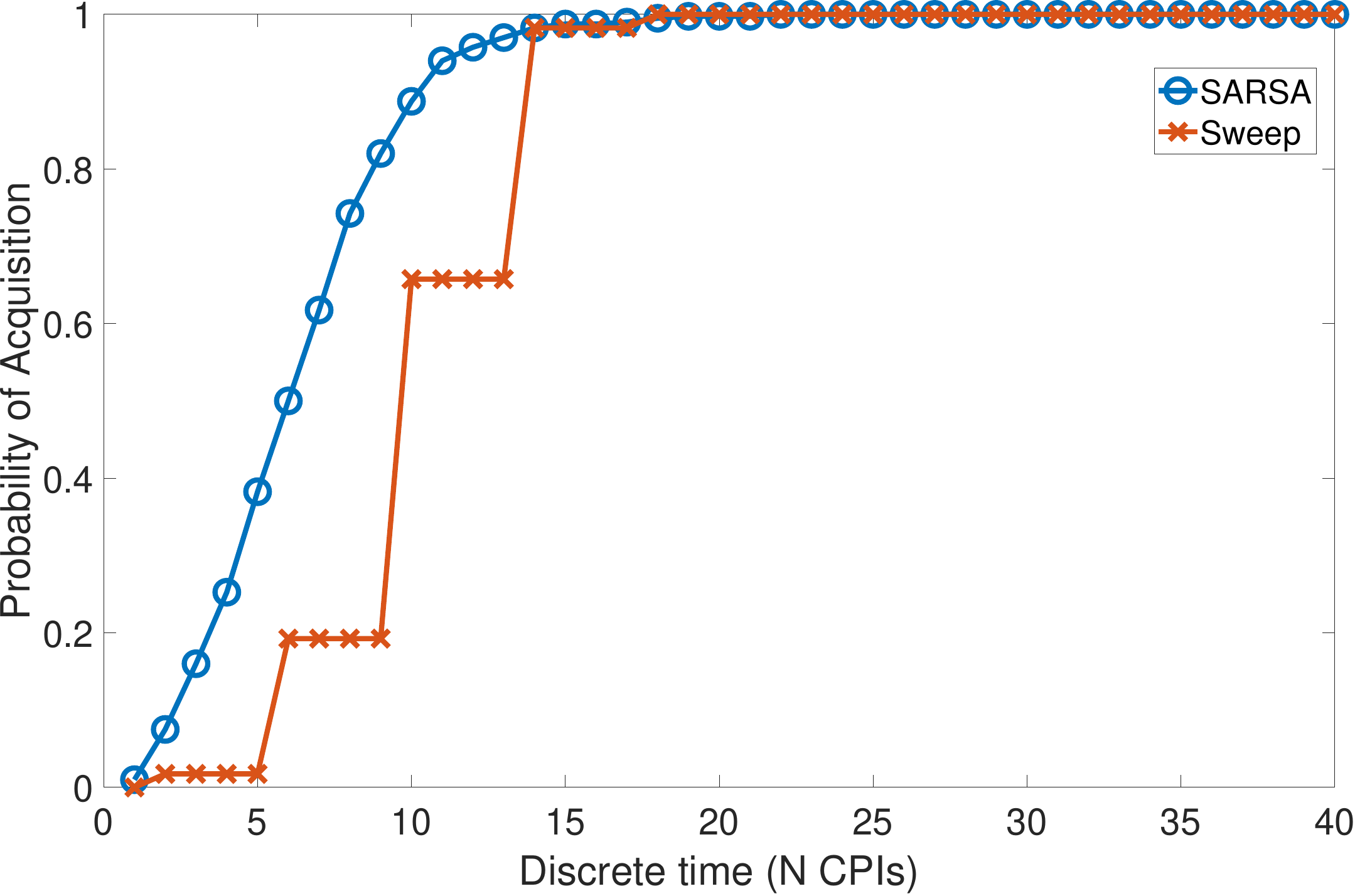}
    \caption{Scenario 2: MIMO (${\rm M_RN_T = 100}$) Probability of Acquisition for the Target in Bin 18}
    \label{fig:Pa:T1}
\end{figure}

\subsubsection{Scenario 3}
We first consider decentralized detection with two radars. There are two targets; one in bin 7, and one in bin 16. The SNR curve of each target as seen by the radars is described by Fig.~\ref{fig:s3:SNRInd}. Practically, one can think of this scenario as representing two radars which face each other where the two targets are moving back and forth between each radar within their respective angle bins.

\begin{figure}[H]
    \centering
    \includegraphics[width=0.95\linewidth]{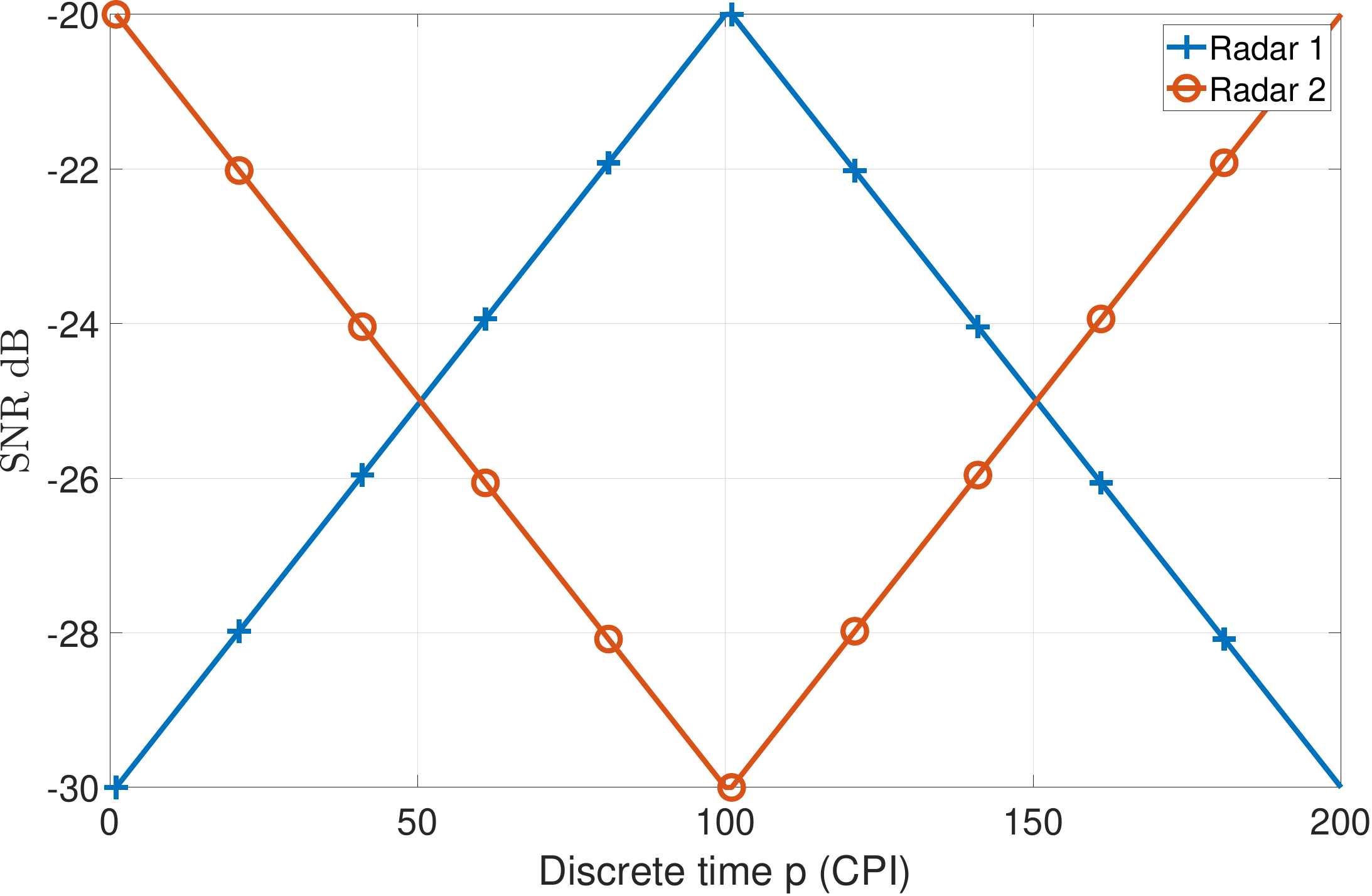}
    \caption{SNR of Target 1 and Target 2 Seen by Radar 1}
    \label{fig:s3:SNRInd}
\end{figure}

Detection performance of the system is shown for the target in angle bin 7 in Fig.~\ref{fig:s3:MIMOT1} and for the target in angle bin 16 in Fig.~\ref{fig:s3:MIMOT2}. We can immediately see that by adding just one additional radar, detection performance drastically improves from Fig.~\ref{fig:s2:MIMOT1} and Fig.~\ref{fig:s2:MIMOT2}. This can be expected as spatial diversity provides an SNR improvement; furthermore, this is a realistic scenario as almost surely all radars in the network will not see the same target SNR.

\begin{figure}[H]
    \centering
    \includegraphics[width=0.95\linewidth]{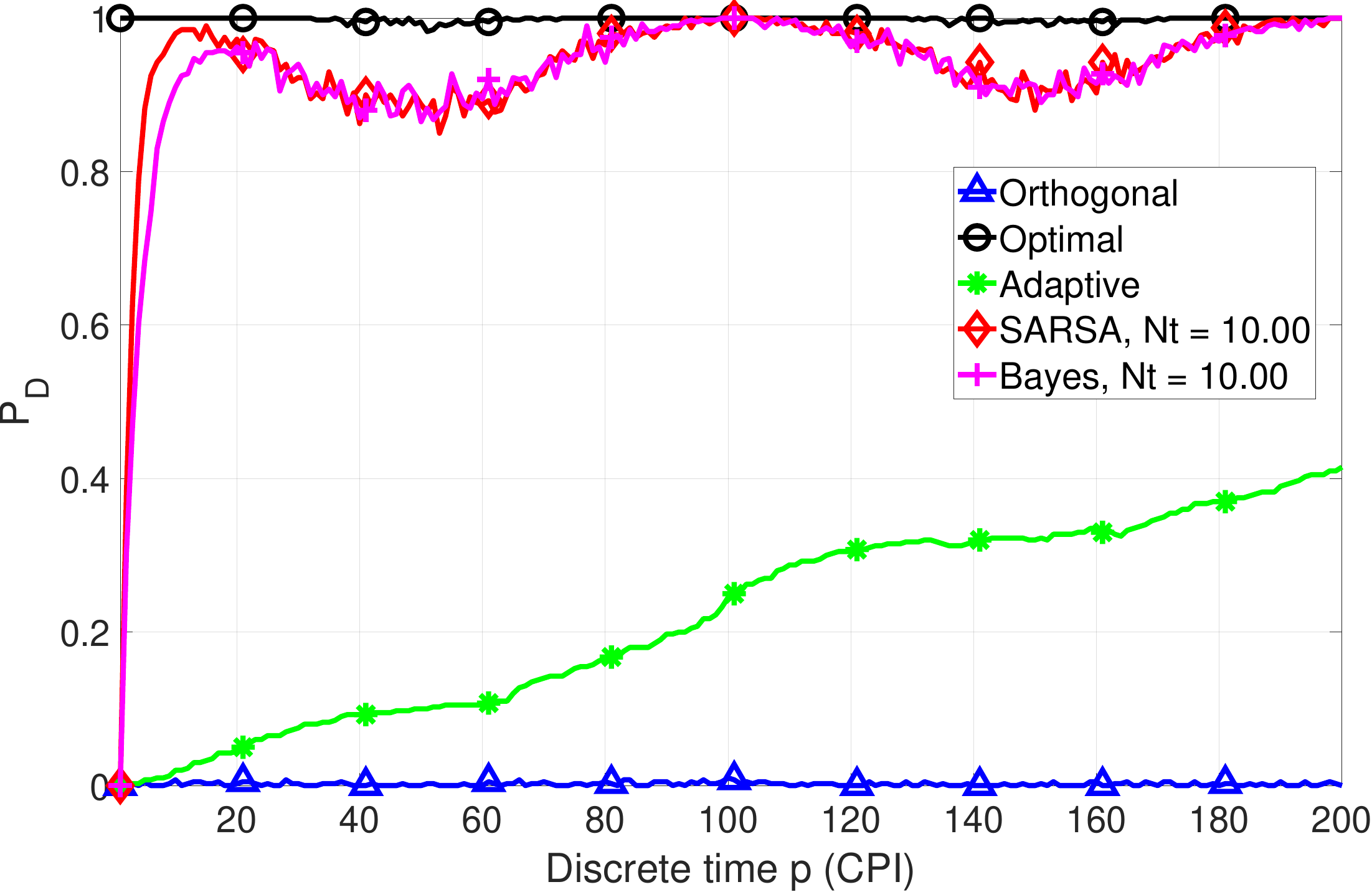}
    \caption{Scenario 3: MIMO (${\rm M_RN_T = 100}$) Decentralized Detection Performance for the Target in Bin 7}
    \label{fig:s3:MIMOT1}
\end{figure}
\begin{figure}[H]
    \centering
    \includegraphics[width=0.95\linewidth]{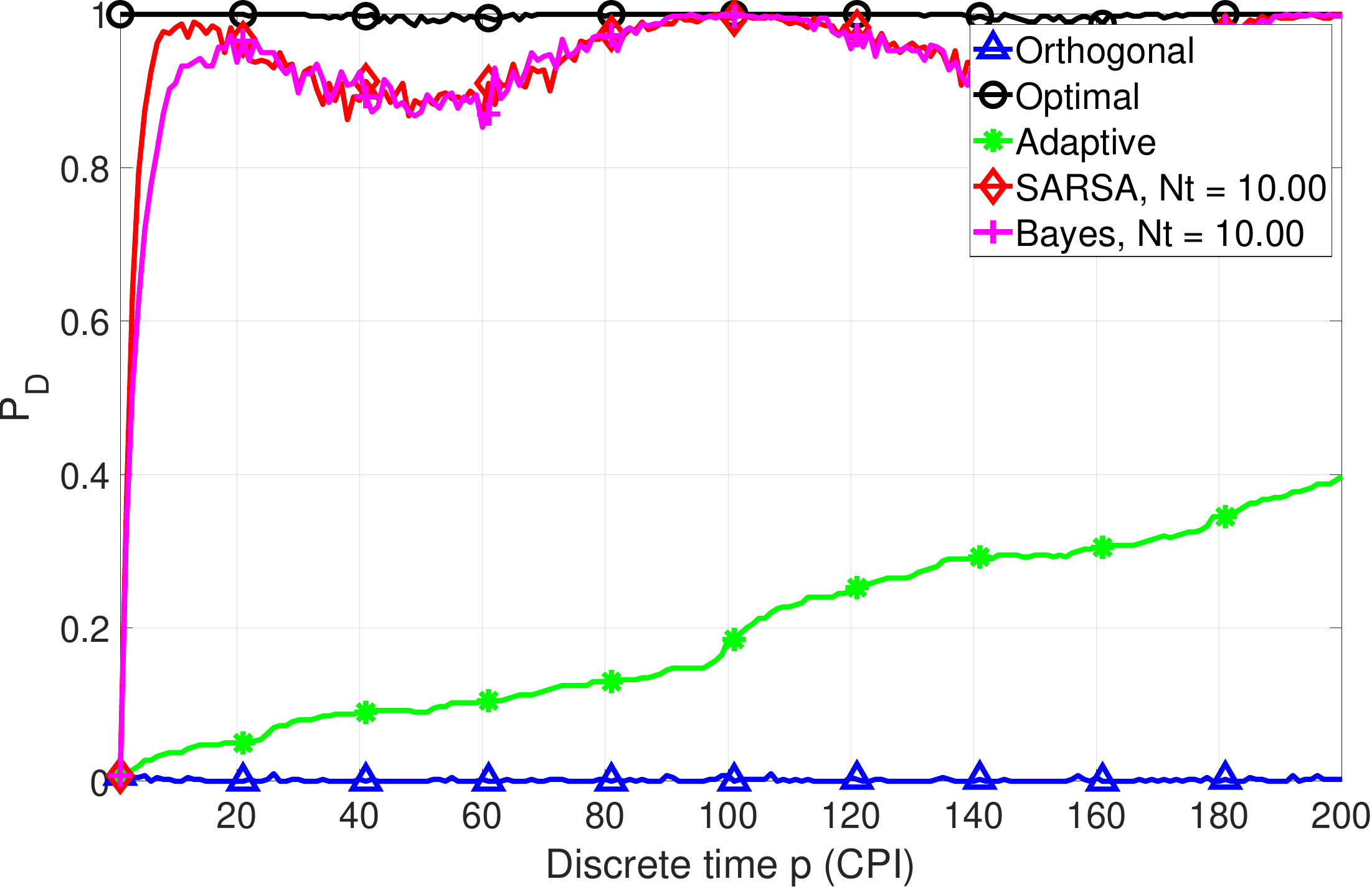}
    \caption{Scenario 3: MIMO (${\rm M_RN_T = 100}$) Decentralized Detection Performance for the Target in Bin 16}
    \label{fig:s3:MIMOT2}
\end{figure}

\subsubsection{Scenario 4}
We now consider Scenario 3 with centralized detection rather than decentralized to show the performance gain. Detection performance for each target is shown in Fig.~\ref{fig:s4:MIMOT1} and Fig.~\ref{fig:s4:MIMOT2}. We can see that just by using centralized combination versus decentralized in the same exact setting, detection performance is improved for both targets. Of course, this comes with the tradeoff of having to transmit the entire received signal from each radar to a central node, rather than just the individual detection statistics. This will result in significantly higher network throughput and could result in processing delays. 
\begin{figure}[H]
    \centering
    \includegraphics[width=0.95\linewidth]{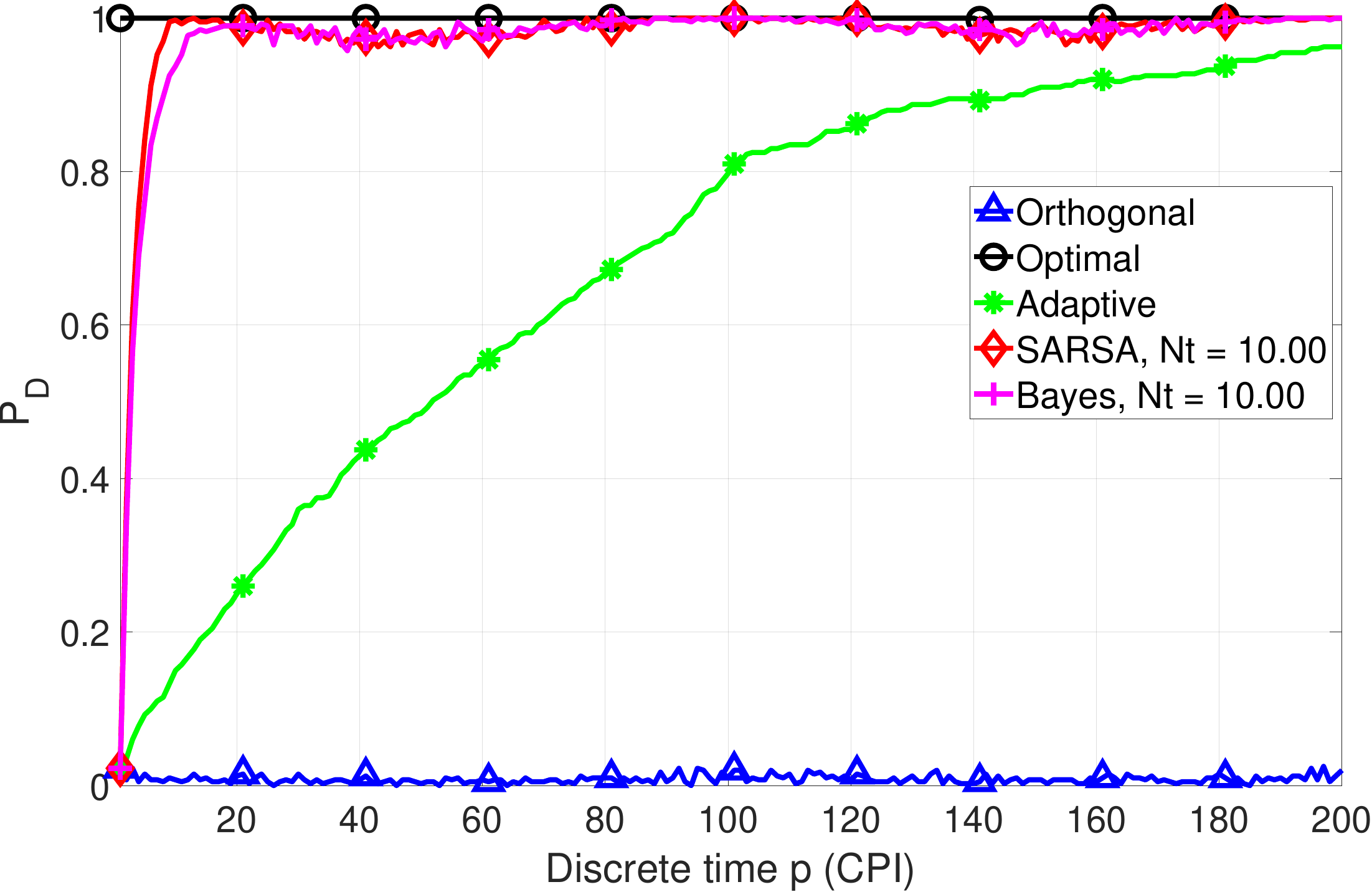}
    \caption{Scenario 4: MIMO (${\rm M_RN_T = 100}$) Centralized Detection Performance for the Target in Bin 7}
    \label{fig:s4:MIMOT1}
\end{figure}
\begin{figure}[H]
    \centering
    \includegraphics[width=0.95\linewidth]{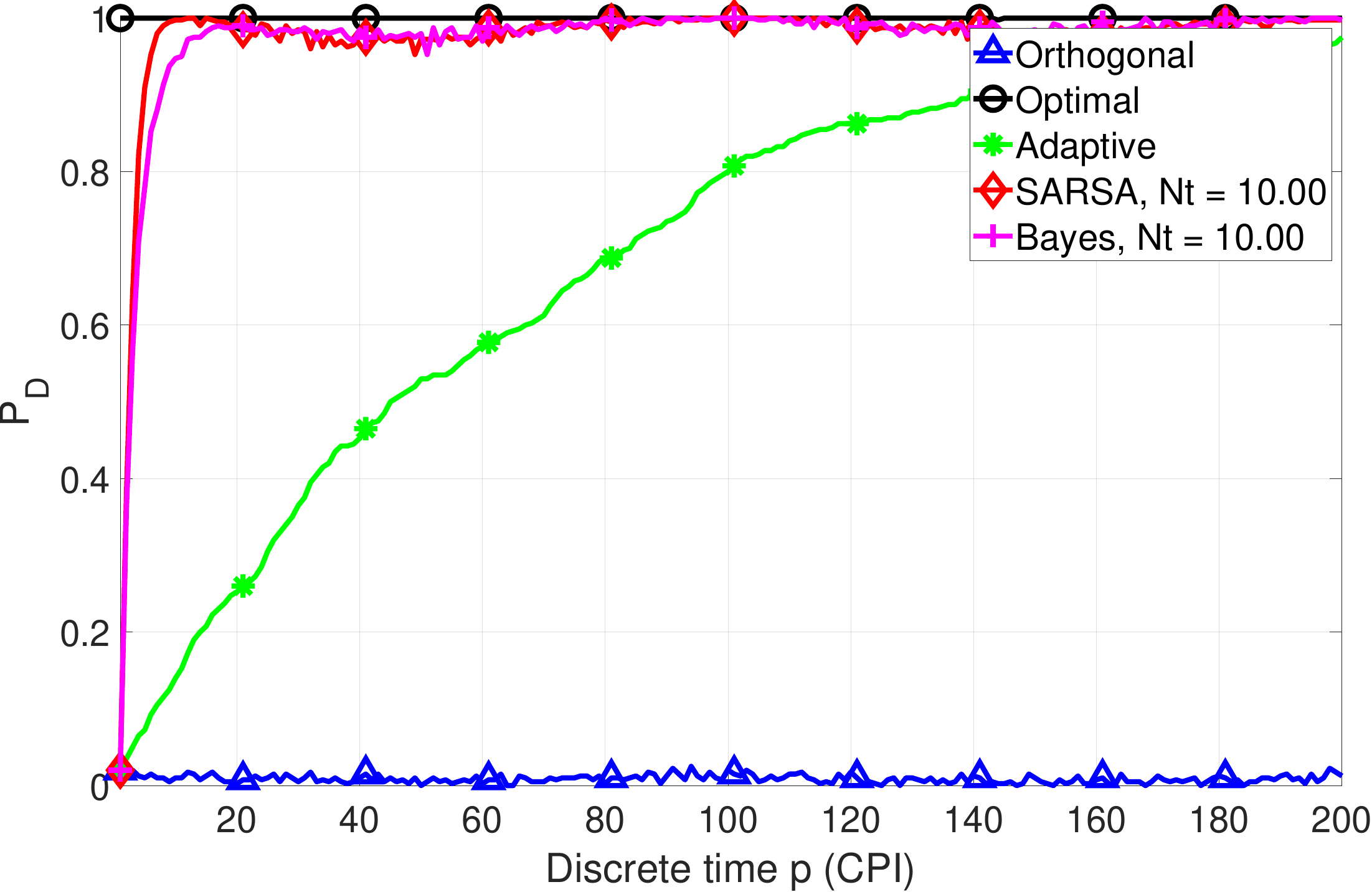}
    \caption{Scenario 4: MIMO (${\rm M_RN_T = 100}$) Centralized Detection Performance for the Target in Bin 16}
    \label{fig:s4:MIMOT2}
\end{figure}

\subsubsection{Scenario 5}
In our previous scenarios for radar networks (Scenarios 3 and 4), we have considered a network of only two radars. While this is sufficient for demonstrating performance improvements, only having two radars does not quite capture the spirit of a radar network. Thus, in this scenario we show centralized detection performance with three radars in a similar setting to Scenario 4. Both targets as seen by the third radar follow an SNR curve described in Fig.~\ref{fig:s5:SNR}. The SNR curve for the other two radars is exactly as previously described in Scenario 4. 

\begin{figure}[H]
    \centering
    \includegraphics[width=0.95\linewidth]{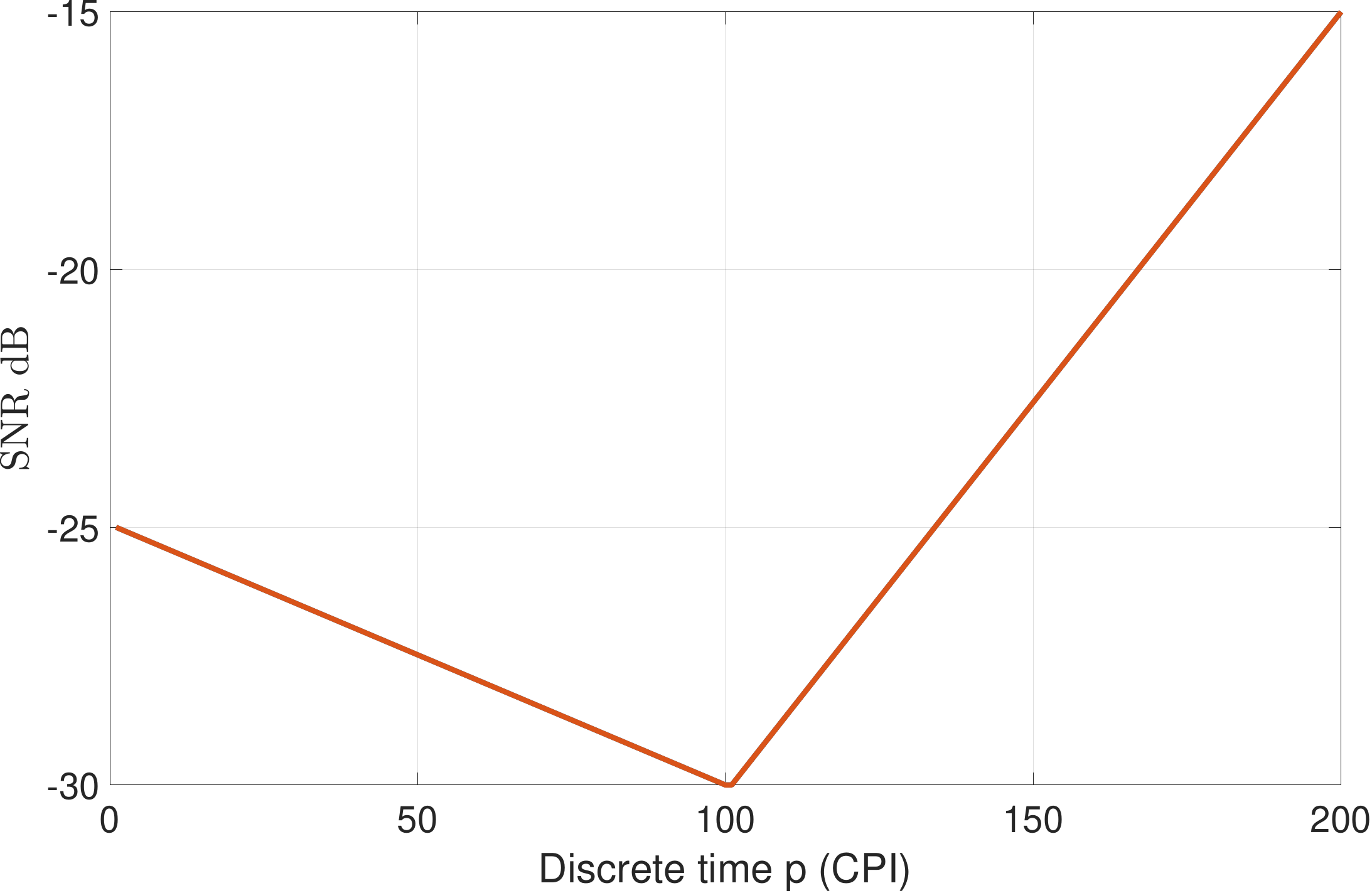}
    \caption{SNR of Third Radar for Both Targets}
    \label{fig:s5:SNR}
\end{figure}

By comparing Fig.~\ref{fig:s5:MIMOT1} and Fig.~\ref{fig:s5:MIMOT2} to Fig.~\ref{fig:s4:MIMOT1} and Fig.~\ref{fig:s4:MIMOT2}, we can see that detection performance has improved. While it may be difficult to tell by looking at the RL based algorithms, we can see that even the orthogonal detection algorithm has vastly improved in performance. This suggests that the overall SNR gain by using centralized detection with additional radars is significant. 

\begin{figure}[H]
    \centering
    \includegraphics[width=0.95\linewidth]{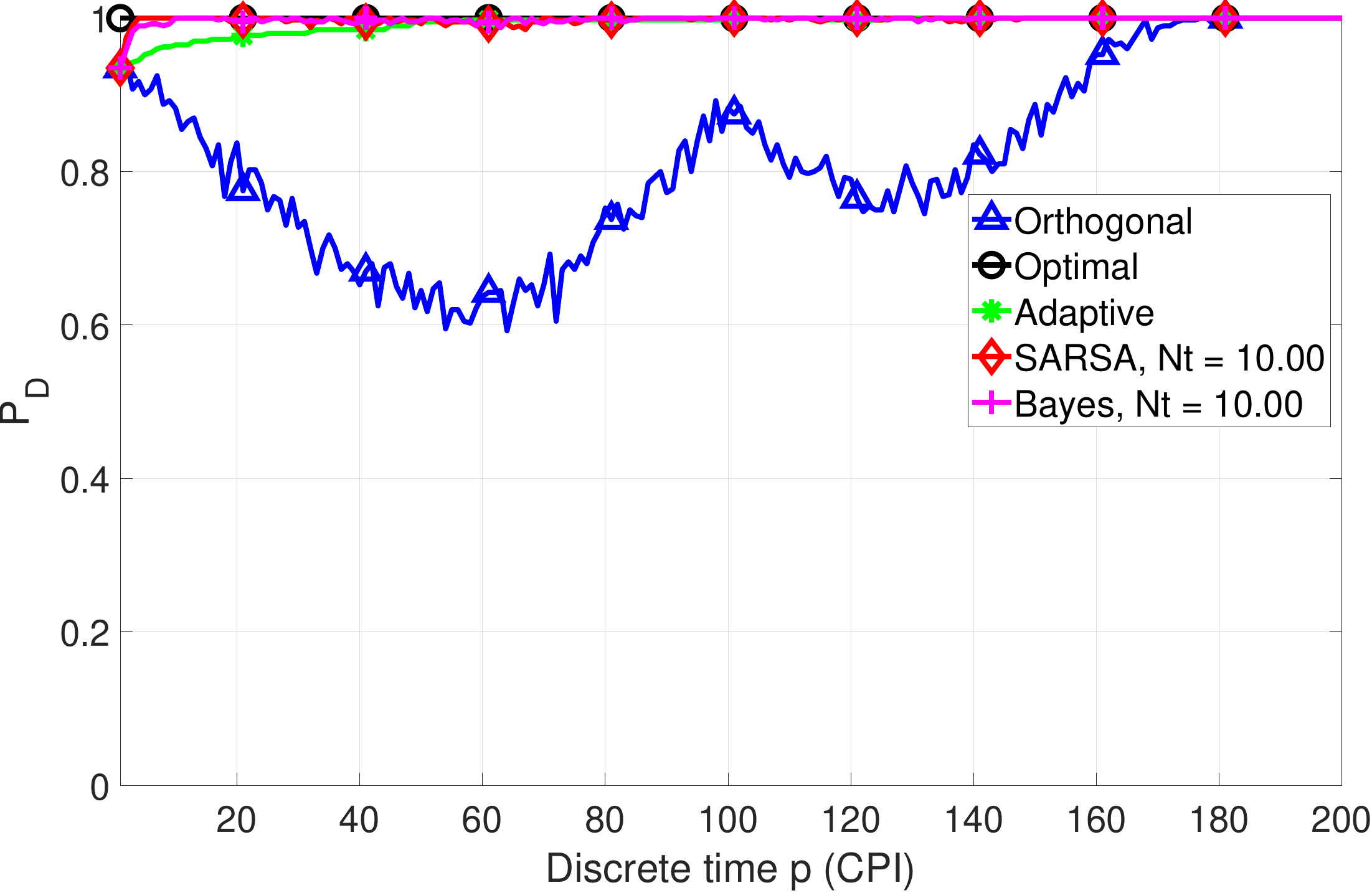}
    \caption{Scenario 5: MIMO (${\rm M_RN_T = 100}$) Centralized Detection Performance for the Target in Bin 7}
    \label{fig:s5:MIMOT1}
\end{figure}
\begin{figure}[H]
    \centering
    \includegraphics[width=0.95\linewidth]{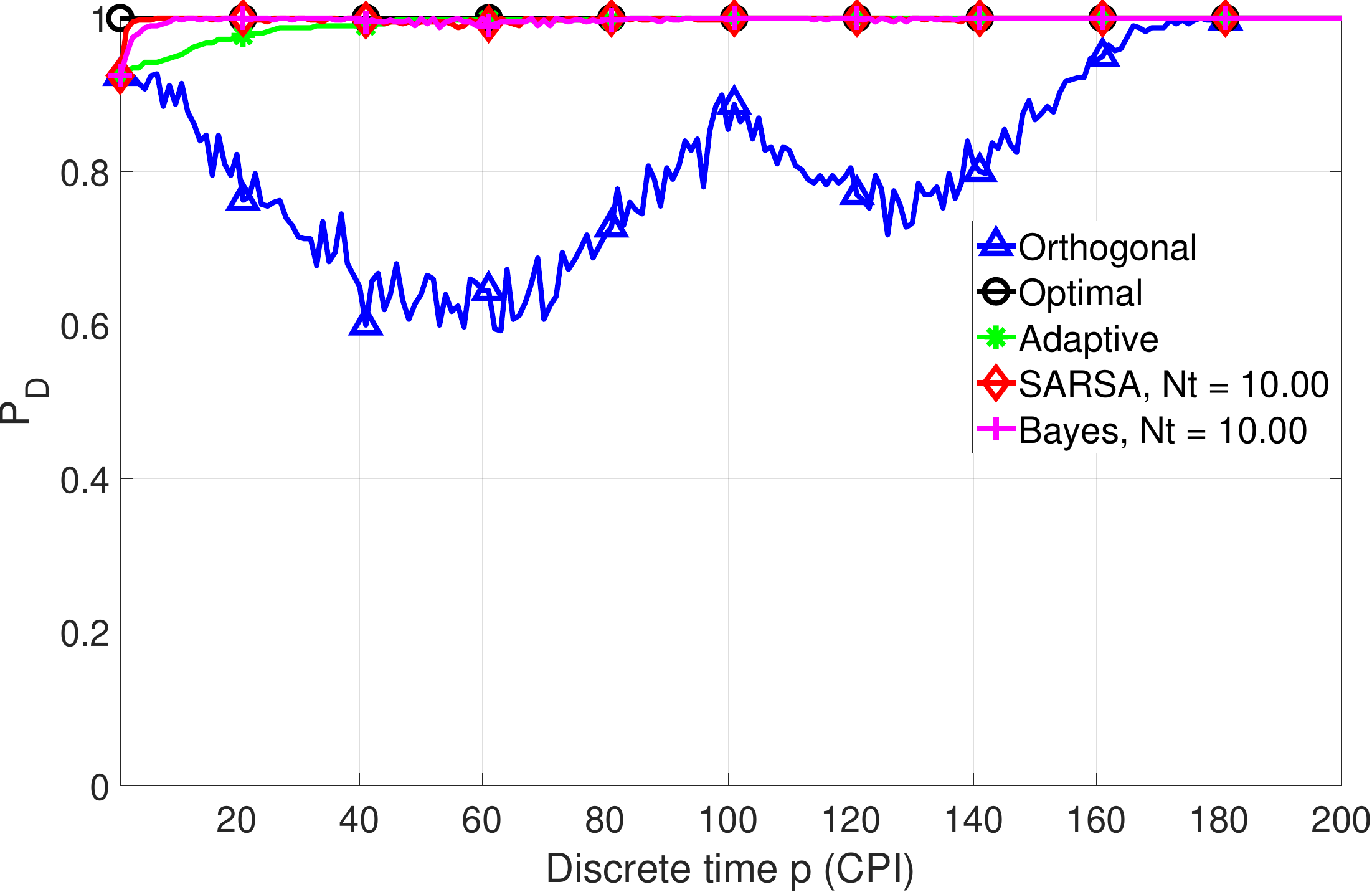}
    \caption{Scenario 5: MIMO (${\rm M_RN_T = 100}$) Centralized Detection Performance for the Target in Bin 16}
    \label{fig:s5:MIMOT2}
\end{figure}

\section{Conclusion}
In this work, we have presented a CFAR detection statistic which is robust to unknown disturbance distributions for co-located MIMO radar. We have described necessary conditions under which this detection statistic is valid, and we have shown that there exist 2-D AR processes which follow these conditions. We have also shown that the use of reinforcement learning greatly improves multi-target detection performance using this test statistic. Finally, we have generalized these aforementioned results to CRNs. We have developed both decentralized and centralized signal fusion techniques which provide a robust detection statistic. Through simulation studies, we have shown the significant detection performance improvements which are a result of a network of radars which work collaboratively. 

Future work should focus on relaxing the necessary conditions outlined for our robust detection statistic to be valid. Furthermore, future work should focus on the joint multi-target detection and tracking problem both with CR and CRNs. Finally, future work could also explore radar node placement such that the condition of a bijective mapping between the angle bins of radars is guaranteed, and performance of CRNs without perfect synchronization could be analyzed.


%





\ifCLASSOPTIONcaptionsoff
  \newpage
\fi



\bibliographystyle{IEEEtran}
\bibliography{bibtex/bib/maindoc}

\end{document}